\documentclass[letterpaper,11pt]{article}
\usepackage[toc,page]{appendix}
\usepackage[margin=1in]{geometry}
\usepackage[bookmarks, colorlinks=true, plainpages = false, citecolor = blue,linkcolor=red,urlcolor = blue, filecolor = blue,pagebackref]{hyperref}
\usepackage{url}
\usepackage{amsmath,amsfonts,amsthm,amssymb,bm, verbatim,dsfont,mathtools}
\usepackage{color,graphicx,appendix}
\usepackage{subfigure}
\usepackage{etoolbox}
\usepackage{tikz}
\usepackage{xr,xspace}
\usepackage{todonotes}
\usepackage{paralist}
\usepackage{caption,soul}
\usepackage{algorithm}
\usepackage{algorithmic}
\usepackage{appendix}
\makeatletter

\newtheorem{theorem}{Theorem}
\newtheorem{lemma}{Lemma}
\newtheorem{proposition}{Proposition}
\newtheorem{corollary}{Corollary}
\theoremstyle{definition}
\newtheorem{definition}{Definition}

\newtheorem{remark}{Remark}
\newtheorem{assumption}{Assumption}

\newcommand{\argmin}{\mathop{\arg\min}}

\usepackage{xspace,prettyref}

\newcommand{\diverge}{\to\infty}

\newcommand{\iiddistr}{{\stackrel{\text{\iid}}{\sim}}}

\newcommand{\reals}{{\mathbb{R}}}

\newcommand{\identity}{\mathbf I}

\newcommand{\expect}[1]{\mathbb{E}\left[ #1 \right]}

\newcommand{\prob}[1]{ \mathbb{P}\left\{ #1 \right\} }

\newcommand{\toas}{\xrightarrow{{\rm a.s.}}}

\newcommand{\Binom}{{\rm Binom}}

\newcommand{\eg}{e.g.\xspace}
\newcommand{\ie}{i.e.\xspace}
\newcommand{\iid}{i.i.d.\xspace}
\newrefformat{eq}{(\ref{#1})}
\newrefformat{chap}{Chapter~\ref{#1}}
\newrefformat{sec}{Section~\ref{#1}}
\newrefformat{alg}{Algorithm~\ref{#1}}
\newrefformat{fig}{Fig.~\ref{#1}}
\newrefformat{tab}{Table~\ref{#1}}
\newrefformat{rmk}{Remark~\ref{#1}}
\newrefformat{clm}{Claim~\ref{#1}}
\newrefformat{def}{Definition~\ref{#1}}
\newrefformat{cor}{Corollary~\ref{#1}}
\newrefformat{lmm}{Lemma~\ref{#1}}
\newrefformat{prop}{Proposition~\ref{#1}}
\newrefformat{app}{Appendix~\ref{#1}}
\newrefformat{hyp}{Hypothesis~\ref{#1}}
\newrefformat{thm}{Theorem~\ref{#1}}
\newrefformat{ass}{Assumption~\ref{#1}}

\newcommand{\pth}[1]{\left( #1 \right)}

\newcommand{\iprod}[2]{\left \langle #1, #2 \right\rangle}
\newcommand{\Iprod}[2]{\langle #1, #2 \rangle}
\newcommand{\indc}[1]{{\mathbf{1}_{\left\{{#1}\right\}}}}

\newcommand{\calA}{{\mathcal{A}}}
\newcommand{\calB}{{\mathcal{B}}}

\newcommand{\calE}{{\mathcal{E}}}
\newcommand{\calF}{{\mathcal{F}}}

\newcommand{\calN}{{\mathcal{N}}}

\newcommand{\calS}{{\mathcal{S}}}

\newcommand{\calV}{{\mathcal{V}}}

\newcommand{\calX}{{\mathcal{X}}}

\newcommand{\median}{\mathsf{med}}

\renewcommand{\hat}{\widehat}
\renewcommand{\tilde}{\widetilde}

\newcommand{\vct}[1]{\bm{#1}}

\begin{document}

\title{Distributed Statistical Machine Learning in Adversarial Settings:
Byzantine Gradient Descent}

\date{\today}
\author{
Yudong Chen \\
Cornell \\
{yudong.chen@cornell.edu }
\and
Lili Su \\
UIUC\\
{lilisu3@illinois.edu}
\and
Jiaming Xu  \\
Purdue\\
{xu972@purdue.edu}}

\maketitle

\begin{abstract}

We consider the distributed statistical learning problem over decentralized systems that are prone to adversarial attacks. This setup arises in many practical applications, including Google's {\em Federated Learning}. Formally, we focus on a decentralized system that consists of a parameter server and $m$ working machines; each working machine keeps $N/m$ data samples, where $N$ is the total number of samples. In each iteration, up to $q$ of the $m$ working machines suffer Byzantine faults -- a faulty machine in the given iteration behaves arbitrarily badly against the system and has complete knowledge of the system. 
Additionally, the sets of faulty machines may be different across iterations.
%
Our goal is to design robust algorithms such that the system can learn the underlying true parameter,
which is of dimension $d$, despite the interruption of the Byzantine attacks.

In this paper, based on the \emph{geometric median of means} of the gradients, we propose a simple variant of the classical gradient descent method.
We show that our method can tolerate $q$ Byzantine failures up to $2(1+\epsilon)q \le m$ for an arbitrarily small but fixed constant $\epsilon>0$.
The parameter estimate converges in $O(\log N)$ rounds with an estimation error on the order of $\max\{\sqrt{dq/N}, ~\sqrt{d/N}\}$, which is larger than the
minimax-optimal error rate $\sqrt{d/N}$  in the centralized and failure-free setting by at most a factor of $\sqrt{q}$.
The total computational complexity of our algorithm is of $O((Nd/m) \log N)$ at each working machine and $O(md + kd \log^3 N)$ at the central server, and the total communication cost is of $O(m d \log N)$.
We further provide an application of our general results to the linear regression problem.

A key challenge arises in the above problem is that Byzantine failures create arbitrary and unspecified dependency among the iterations and the aggregated gradients.
To handle this issue in the analysis, we prove that the aggregated gradient, as a function of model parameter, converges \emph{uniformly} to the true gradient function.

\end{abstract}

\section{Introduction}
\label{intro}
%
Distributed machine learning has emerged as an attractive solution to large-scale problems and received intensive attention~\cite{boyd2011distributed,jordan2016communication,moritz2015sparknet,provost1996scaling,dean2008mapreduce,low2012distributed}.
In this setting, the data samples or/and computation are distributed across multiple machines, which are programmed to collaboratively learn a model. 
Many efficient distributed machine learning algorithms \cite{boyd2011distributed,jordan2016communication} and system implementations~\cite{moritz2015sparknet,provost1996scaling,dean2008mapreduce,low2012distributed} have been proposed and studied.
Prior work mostly focuses on the traditional ``training within cloud'' framework where
the model training process is carried out within the cloud infrastructures.
In this framework,  distributed machine learning is secured via system architectures, hardware devices, and monitoring \cite{kaufman2002network,Pfleeger:2002:SC:579149,wang2010privacy}. This framework faces significant privacy risk, as the data
has to be collected from owners and stored within the clouds.  Although a variety of privacy-preserving solutions have been developed \cite{agrawal2000privacy,duchi2013local}, 
privacy breaches occur frequently, with recent examples including iCloud leaks of celebrity photos and PRISM surveillance program.

To address privacy concerns\footnote{We would like to characterize the amount of privacy sacrificed in the Federated Learning paradigm. We leave this characterization as one of our future work.} , a new machine learning paradigm called {\em Federated Learning} was proposed by Google researchers \cite{konevcny2015federated,federatedlearningblog}. It aims at learning an accurate model  without collecting data from owners and storing the data in the cloud.
The training data is kept locally on the owners' computing devices, which are recruited to participate directly in the model training process and hence function as working machines. Google has been intensively testing this new paradigm in their recent projects such as {\em Gboard}~\cite{federatedlearningblog}, the Google Keyboard. Compared to ``training within cloud'', Federated Learning faces the following three key challenges:
\begin{itemize}
\item Security:  The devices of the recruited data owners can be easily reprogrammed and completely controlled by external attackers, and thus behave adversarially.
\item Small local datasets versus high model complexity: While the total number of data samples over all data owners may be large, each individual owner
may keep only a small amount of data, which by itself is insufficient for learning a complex model.

\item Communication constraints:  Data transmission between the recruited devices and the cloud may suffer from high latency and low-throughout. Communication between them is therefore a scarce resource. 

\end{itemize}

In this paper, we address the above challenges by developing a simple variant of the gradient descent method that can
(1) tolerate the arbitrary and adversarial failures, (2) accurately learn a highly complex model with low local data volume, and
(3) converge exponentially fast using logarithmic communication rounds. Since gradient descent
algorithms  are well-adopted in existing implementations and applications, our proposed method only requires
a small amount of modification of existing codes.

\vskip 0.8\baselineskip

Note that there are many other challenges besides what are listed here, including unevenly distributed training data, intermittent availability of mobile phones, etc. These challenges will be addressed in future work.

\subsection{Learning Goals}
To formally study the distributed machine learning problem in adversarial settings, we consider a standard statistical learning setup, where the
data is generated probabilistically from an unknown distribution and the true model is parameterized by a vector.
More specifically, let $X \in \calX $ be the input data generated according to some {\em unknown} distribution $\mu$. 
Let $\Theta \subset \reals^d$ be the set of all choices of model parameters.
We consider a loss function $f: \calX \times \Theta \to \reals$, where $f(x, \theta)$ measures the risk induced by a realization $x $ of the data under the model parameter choice $\theta$. A classical example
is linear regression, where $ x = (w, y) \in \reals^{d} \times \reals $ is the feature-response pair and $f(x, \theta) = \frac{1}{2}
\left( \iprod{w}{\theta} - y \right)^2 $ is the usual squared loss.

We are interested in learning the model choice $\theta^*$ that minimizes the {\em population risk}, i.e., 
\begin{align}
\label{eq:min_pop_risk}
\theta^* \in \arg \min_{\theta \in \Theta} F(\theta) \triangleq \expect{f(X, \theta)},
\end{align}
assuming that $\expect{ f( X, \theta)}$ is well defined over $\Theta$.\footnote{For example, if $\expect{| f( X, \theta)|}$ is finite for every $\theta \in \Theta$, the population risk $\expect{f(X, \theta)}$ is well defined.} The model choice $\theta^*$ is optimal in the sense that it minimizes the average risk to pay if the
model chosen is used for prediction in the future with a fresh random data sample.

When $\mu$---the distribution of $X$---is known, which is rarely the case in practice, the population risk can be evaluated exactly, and $\theta^*$ can be computed by solving the minimization problem in \eqref{eq:min_pop_risk}. We focus on the more realistic scenario where $\mu$ is  {\em unknown} but there exist $N$ independently and identically distributed data samples $X_i \iiddistr \mu$ for $i = 1,\ldots, N$. Note that estimating $ \theta^* $ using finitely many data samples will always have a \emph{statistical error} due to the randomness in the data, even in the centralized, failure-free setting.  Our results account for this effect.

\subsection{System Model}
We focus on solving the above statistical learning problem over decentralized systems that are prone to adversarial attacks. Specifically, the system of interest consists of a parameter server\footnote{Note that, due to communication bandwidth constraints, practical systems use multiple networked parameter servers. In this paper, for ease of explanation, we assume there is only one parameter server. Fortunately, as can be seen from our algorithm descriptions and our detailed correctness analysis, the proposed algorithm also works for the multi-server setting.} and $m$ working machines. In the example of Federated Learning,
the parameter server represents the cloud, and the $m$ working machines correspond to $m$ data owners' computing devices.

We assume that the $N$ data samples are distributed evenly across the $m$ working machines. In particular, each working machine $i$ keeps a subset $\calS_i$ of the data, where  $\calS_i\cap \calS_j=\emptyset$ and $|\calS_i|=N/m.$ Note that this is a simplifying assumption of the data imbalance in Federated Learning. Nevertheless, our results can be extended to the heterogeneous data sizes setting when the data sizes are of the same order.
We further assume that the parameter server can communicate with all working machines in synchronous communication rounds, and leave the asynchronous setting as future directions. 

Among the $m$ working machines, we assume that up to $q$ of them can suffer Byzantine failures and thus behave maliciously; for example, they may be reprogrammed and completely controlled by the system attacker. 
We assume the parameter server knows $q$ -- as $q$ can be estimated from the existing system failures statistics.
The set of Byzantine machines can {\em change} between communication rounds; the system attacker can choose different sets of machines to control across communication rounds.
Byzantine faulty machines are assumed to have {\em complete knowledge} of the system, including the total number of working machines $m$, all $N$ data samples over the whole system, the programs that the working machines are supposed to run, the program run by the parameter server, and the realization of the random bits generated by the parameter server. Moreover, Byzantine machines can collude ~\cite{Lynch:1996:DA:2821576}. The only constraint is that these machines cannot corrupt the local data on working machines --- but they can lie when communicating with the server. In fact, our main results show that our proposed algorithm still works when at most $q$ different machines with local data corrupted during its execution. 

We remark that Byzantine failures are used to capture the unpredictability of extremely large system that consists of heterogeneous processes, as is the case with Federated Learning. The arbitrary behavior of Byzantine machines creates unspecified dependency across communication rounds --- a key challenge in our algorithm design and convergence analysis. In this paper, we use {\em rounds} and {\em iterations} interchangeably.


\subsection{Existing Distributed Machine Learning Algorithms}
There are three popular classes of existing distributed machine learning algorithms
in terms of their communication rounds. \\

\noindent {\bf SGD: } On one end of the spectrum lies the {\em Stochastic Gradient Descent (SGD)} algorithm. Using this algorithm, the parameter server receives, in each iteration, a gradient computed at a single data sample from one  working machine, and uses it to perform one gradient descent step. Even when the population risk $F$ is strongly convex,  the convergence rate of SGD is only $O(1/t)$ with $t$  iterations. This is much slower than the exponential (geometric) convergence of standard gradient descent.
Therefore, SGD requires a large number of communication rounds, which could be costly. Indeed, it has been demonstrated in~\cite{federatedlearningblog} that  SGD
has 10-100 times higher communication cost than standard gradient descent, and
is therefore inadequate for scenarios with scarce communication bandwidth. \\

\noindent {\bf One-Shot Aggregation: }  On the other end of the spectrum, using a {\em One-Shot Aggregation} method, each working machine computes an estimate of the model parameter using only its local data and reports it to the server, which then aggregates all the estimates reported  to obtain a
final estimate~\cite{JMLR:v14:zhang13b,zhang2015divide}. One-shot aggregation method only needs a single round of communication from the working machines to the parameter
server, and thus is communication-efficient. However, it requires $N/m\gg d$ 
so that a coarse parameter estimate can be obtained at each machine. This algorithm is therefore not applicable in scenarios where local data is small in size but the model to learn is of high dimension. \\

\noindent {\bf BGD: }  {\em Batch Gradient Descent (BGD)} lies in between the above two extremes. At each iteration,  the parameter server sends
the current model parameter estimate to all working machines. Each working machine computes the gradient based on all locally available data, and then sends the gradient back to the parameter server. The parameter server averages the received gradients and performs a gradient descent step.
When the population risk $F$ is strongly convex, BGD converges
exponentially fast, and hence requires only a few rounds of communication. BGD also works in the scenarios with
limited local data, \ie, $N/m =O(d)$, making it an ideal candidate in Federated Learning.
However, it is sensitive to Byzantine failures;
a single Byzantine
failure at a working machine can completely skew the average value of the gradients received by the parameter server, and thus foils the algorithm.

\subsection{Contributions}
In this paper, we propose a Byzantine gradient descent method.
Specifically,  the parameter server aggregates the local gradients reported by the working machines in three steps: (1) it partitions all the received local gradients into $k$ batches and computes the mean for each batch, (2) it computes the \emph{geometric median} of the $k$ batch means, and (3) it performs a gradient descent step using the geometric median.

We prove that the proposed algorithm can tolerate $q$ Byzantine failures up to $2(1+\epsilon)q \le m$ for an arbitrarily small but fixed constant $\epsilon>0$.
Moreover, the error in estimating the target model parameter $\theta^*$
converges in $\log (N)$ communication rounds  to the order of  $\max\{\sqrt{dq/N},~ \sqrt{d/N}\}$, whereas the minimax-optimal
estimation error rate in the centralized and failure-free setting is
$\sqrt{d/N}$.\footnote{Note that
$\sqrt{d/N}$ is the minimax optimal estimation error rate even in the centralized, failure-free setting
when we would like to estimate a $d$-dimensional unknown parameter without any additional
structure from $N$ i.i.d.\ samples,
see \eg, \cite[Section 3.2]{YW-ITSTATS} for a proof in the special case of Gaussian mean estimation. When there is additional structure, say sparsity, then the $\sqrt{d}$ factor can possibly be improved.} Even in the scarce {\em local} data regime where $N/m = O(d)$, the estimator of our proposed algorithm is still consistent as long as $N/q=\omega(d)$.
The total computational complexity of our algorithm is of $O((N/m) d \log N)$ at each worker and
$O(md + qd \log^3 N)$ at the parameter server, and the total communication cost is of $O(m d \log N)$. Note that the $\sqrt{q}$ factor in our estimation error rate $\max\{\sqrt{dq/N},~ \sqrt{d/N}\}$ may not be fundamental to the problem of learning in adversarial settings. Thus, it may possibly be improved with better algorithms or finer analysis.



A key challenge in our analysis is that there exists complicated probabilistic dependency among the iterates and the aggregated gradients.  Even worse, such dependency cannot be specified due to the arbitrary behavior of the Byzantine machines.
We overcome this  challenge by
proving that the geometric median of means of gradients
\emph{uniformly} converges to the true gradient function $\nabla F(\cdot)$.


\subsection{Outline}
The origination of the paper is as follows. In Section \ref{sec: algorithm and summary}, we present our algorithm, named {\em Byzantine Gradient Descent Method}, and summarize our convergence results. Detailed convergence analysis can be found in Section \ref{sec: analysis}. To illustrate the applicability of our convergence results, we provide a linear regression example in Section \ref{sec: linear regression}.  Related work is discussed in Section \ref{sec: related work}.  Section \ref{sec: discussion} concludes the paper, and presents several interesting future directions.

\section{Algorithms and Summary of Convergence Results}
\label{sec: algorithm and summary}
In this section, we present our distributed statistical machine learning algorithm, named {\em Byzantine Gradient Descent Method}, and briefly summarize our convergence results on its performance. 


\subsection{Byzantine Gradient Descent Method}
Recall that our fundamental goal is to learn the optimal model choice $\theta^*$ defined in \eqref{eq:min_pop_risk}.  We make the following standard assumption \cite{boyd2011distributed} so that the minimization problem in \eqref{eq:min_pop_risk} can be solved efficiently (exponentially fast) in the ideal case when the population risk function $F$ is known exactly, i.e., the distribution $\mu$ is known.
\begin{assumption}\label{ass:pop_risk_smooth}
The population risk function $F: \Theta \to \reals$ is $L$-strongly convex, and differentiable over $\Theta$ with $M$-Lipschitz gradient. That is, for all $\theta,\theta' \in \Theta$,
\begin{align*}
&F(\theta' ) \ge F(\theta) + \iprod{\nabla F(\theta)}{\theta'-\theta} + \frac{L}{2} \| \theta'-\theta\|^2,
\end{align*}
and
\begin{align*}
\| \nabla F(\theta) - \nabla F(\theta') \| \le M \| \theta-\theta'\|.
\end{align*}
\end{assumption}
Under Assumption \ref{ass:pop_risk_smooth}, it is well-known \cite{boyd2004convex} that using the standard gradient descent update %
\begin{align}
\theta_t= \theta_{t-1} -\eta \times \nabla F(\theta_{t-1}),\label{eq:pop_gd o}
\end{align}
where $\eta$ is some fixed stepsize, $\theta_t$ approaches $\theta^*$ exponentially fast. In particular, choosing $\eta=L/(2M^2)$, it holds that
$$
\| \theta_t-\theta^* \| \le \pth{1-\pth{\frac{L}{2M}}^2}^{t/2} \| \theta_0 -\theta^*\|.
$$
Nevertheless, when the distribution $\mu$ is unknown, as assumed in this paper, the population gradient $ \nabla F$ can only be approximated using sample gradients, if they exist.

Recall that each working machine $j$ (can possibly be Byzantine) keeps a very small set of data $\calS_j$ with $|\calS_j|=N/m$. Define the local empirical risk function, denoted by $\bar{f}^{(j)}: \Theta \to \reals$, as follows:
\begin{align}
\label{local obj}
\bar{f}^{(j)} ( \theta) \triangleq  \frac{1}{ |\calS_j| } \sum_{ i \in \calS_j } f(X_i, \theta), ~~~ \forall \, \theta\in \Theta.
\end{align}
Notice that $\bar{f}^{(j)} (\cdot)$ is a function of data samples $\calS_j$ stored at machine $j$. Hence $\bar{f}^{(j)} (\cdot)$ is random. Although Byzantine machines can send arbitrarily malicious messages to the parameter server, they are unable to corrupt the local stored data. Thus, the local risk function $\bar{f}^{(j)}(\cdot)$ is well-defined for all $j$, including the Byzantine machines. With a bit of abuse of notation, we let
\begin{align*}
\vct{\bar{f}}(\theta)\triangleq \left(\bar{f}^{(1)} (\theta), \ldots, \bar{f}^{(m)}(\theta) \right)
\end{align*}
be the vector that stacks the values of the $m$ local functions evaluated at $\theta$.
For any $x\in \calX$, we assume that $f(x, .): \Theta \to \reals$  is differentiable. When there is no confusion, we write $\nabla_\theta f(x, \theta)$ -- the gradient of function $f(x, \cdot)$ evaluated at
$\theta$ -- simply as $\nabla f(x, \theta)$.

It is well-known that the average of the local gradients can be viewed as an approximation of the population gradient $\nabla F(\cdot)$. In particular, for a fixed $\theta$, as $N \diverge$
\begin{align}
\label{gradient app ave}
\frac{1}{m} \sum_{j=1}^m \nabla \bar{f}^{(j)} ( \theta) =  \frac{1}{N} \sum_{i=1}^N \nabla f(X_i, \theta)  ~ \toas ~  \nabla F(\theta).
\end{align}
Batch Gradient Descent relies on this observation. However, this method is sensitive to Byzantine failures as we explain next.

\paragraph{\bf  Batch Gradient Descent}
We describe the {\em Batch Gradient Descent (BGD)} in Algorithm \ref{alg:SGD}.
We initialize $\theta_0$ to be some arbitrary value in $\Theta$ for simplicity.
In practice, there are standard guides in choosing the initial point \cite{sherali1998network}. In round $t\ge 1$, the parameter server sends the current model parameter estimator $\theta_{t-1}$ to all working machines.
Each working machine $j$ computes
the gradient $\nabla \bar{f}^{(j)}(\theta_{t-1})$ and sends $\nabla \bar{f}^{(j)}(\theta_{t-1})$ back
to the parameter server. Note that Byzantine machines may not follow the codes in Algorithm \ref{alg:SGD}. 
Instead of the true local gradients, Byzantine machines can report arbitrarily malicious messages or no message to the server.
If the server does not receive any message from a working machine,
then that machine must be Byzantine faulty.
In that case, the server sets $g^{(j)}_t(\theta_{t-1})$ to some arbitrary value.
Precisely, let $\calB_t$ denote the set of Byzantine machines at round $t$ in a given execution. 
The message received from machine $j$, denoted by $g^{(j)}_t (\theta_{t-1})$, can be described as
\begin{align}
\label{received gradients}
g^{(j)}_t (\theta_{t-1}) = \begin{cases}
\nabla \bar{f}^{(j)}(\theta_{t-1}) & \text{ if } j \notin \calB_t \\
\star & \text{ o.w. },
\end{cases}
\end{align}
where, with a bit of abuse of notation, $\star$ denotes the arbitrary message whose value may be different across Byzantine machines, iterations, executions, etc. In step 3, the parameter server averages the received $g^{(j)}_t (\theta_{t-1})$ and updates $\theta_{t}$ using a gradient descent step.
\begin{algorithm}[htb]
\caption{Standard Gradient Descent: Iteration $t\ge 1$}\label{alg:SGD}
\begin{center}
{\bf \color{blue}{\em Parameter server:}}
\end{center}

\begin{algorithmic}[1]
\STATE {\em Initialize:} Let $\theta_0$ be an arbitrary point in $\Theta$.
\STATE Broadcast the current model parameter estimator $\theta_{t-1}$ to all working machines;
\STATE Wait to receive all the gradients reported by the $m$ machines; Let $g^{(j)}_t(\theta_{t-1})$ denote the value received from machine $j$. \\
If no message from machine $j$ is received, set $g^{(j)}_t(\theta_{t-1})$ to be some arbitrary value;
\STATE  Update: $ \theta_{t} ~ \gets ~ \theta_{t-1}   -\eta \times \pth{\frac{1}{m} \sum_{j=1}^m g_t^{(j)}(\theta_{t-1}) }$;
\end{algorithmic}

\vskip \baselineskip

\begin{center}
{\bf \color{blue}{\em Working machine $j$:}}
\end{center}
\begin{algorithmic}[1]
\STATE Compute the gradient $\nabla \bar{f}^{(j)}(\theta_{t-1})$;
\STATE Send $\nabla \bar{f}^{(j)}(\theta_{t-1})$ back to the parameter server;
\end{algorithmic}
\end{algorithm}

Under Assumption \ref{ass:pop_risk_smooth}, when there are no Byzantine machines, it is well-known that BGD converges exponentially fast. However, a single Byzantine failure can completely skew the average value of the gradients received by the parameter server, and thus foils the algorithm. It is still the case even if the parameter server takes an average of
a randomly selected subset of received gradients. This is because a Byzantine machine is assumed to have complete knowledge of the system, including the gradients reported by other machines, and the realization of the random bits generated by the parameter server. \\

%
\paragraph{\bf Robust Gradient Aggregation}
Instead of taking the average of the received gradients
$$g_t^{(1)}(\theta_{t-1}), \cdots, g_t^{(m)}(\theta_{t-1}),$$
we propose a robust way to aggregate the collected gradients. Our aggregation rule is based on the notion of \emph{geometric median}.

Geometric median is a generalization of median in one-dimension to multiple dimensions, and has been widely used in robust
statistics~\cite{mottonen2010asymptotic,milasevic1987uniqueness,kemperman1987median,cardot2013efficient}.
Let $\{ y_1, \ldots, y_n \} \subseteq \reals^d$ be a multi-set of size $n$. The geometric median of $\{ y_1, \ldots, y_n \}$, denoted by $\median\{ y_1, \ldots, y_n \}$, is defined as
\begin{align}
\label{median}
\median \{ y_1, \ldots, y_n \} ~\triangleq ~\argmin_{y\in \reals^d}  \sum_{i=1}^n \| y - y_i \|.
\end{align}
Geometric median is NOT required to lie in $\{ y_1, \ldots, y_n \}$, and is unique unless all the points in $\{ y_1, \ldots, y_n \}$ lie on a line. Note that if the $\ell_2$ norm in \eqref{median} is replaced by the squared $\ell_2$ norm, i.e., $\| \cdot \|^2$, then the minimizer is exactly the average.

In one dimension, median has the following nice robustness property:
if strictly more than $\lfloor n/2 \rfloor$ points are in $[-r,r]$ for some $r\in\reals$,
then the median must be in $[-r,r]$.  
Likewise, in multiple dimensions, geometric
median has similar robust property \cite[Lemma 2.1]{minsker2015geometric}
\cite[Lemma 24]{cohen2016geometric}.
The following lemma shows that a $(1+\gamma)$-
approximate geometric median is also robust. Its
proof is a simple adaptation of the proof of Lemma 24
in \cite{cohen2016geometric} and presented in
Appendix~\ref{pf_geometric_median_robust}.

\begin{lemma}
\label{lmm:geometric_median_robust}
Let $z_1, \ldots, z_n $ denote $n$ points in a Hilbert space.
Let $z_\ast$ denote a $(1+\gamma)$-approximation of
 their geometric median, \ie, $ \sum_{i=1}^n \| z_\ast -z_i \|
 \le (1+\gamma) \min_{z} \sum_{i=1}^n \| z - z_i \| $ for $\gamma>0$.
For any $\alpha \in (0, 1/2)$ and given $r \in \reals$, if
$\sum_{i=1}^n \indc{\|z_i \|  \le r} \ge (1- \alpha) n$,
then
\begin{align*}
 \| z_\ast \|  & \le C_\alpha r + \gamma \frac{\min_{z} \sum_{i=1}^n  \| z - z_i \|  }{(1-2\alpha) n } \le C_\alpha r + \gamma \frac{\max_{1 \le i \le n} \|z_i \| }{1-2\alpha},
 \end{align*}
where
\begin{align}
C_\alpha = \frac{2(1-\alpha)}{1-2\alpha}.
\end{align}
\end{lemma}

The above lemma shows that as long as there are sufficiently many points (majority in terms of fraction) inside the
Euclidean ball of radius $r$ centered at origin, then the geometric median ($\gamma=0$)
must lie in the Euclidean ball blowed up by a constant factor only. Intuitively, geometric median can be viewed as an  aggregated center of a set based on majority vote.
Note that the exact geometric median may not be computed efficiently in practice. The above lemma further shows that
$(1+\gamma)$-approximate geometric median also lies in the Euclidean ball blowed up by a constant factor plus
a deviation term proportional to $\gamma$ and $\max_i \| z_i\|$.

Let
$
 \vct{g}_t (\theta_{t-1} )= \left( g^{(1)}_t (\theta_{t-1} ), \ldots, g^{(m)}_t(\theta_{t-1} ) \right)
$
be the vector that stacks the gradients received by the parameter server at iteration $t$.
Let $k$ be an integer which divides $m$ and let $b=m/k$ denote the batch size. In our proposed robust gradient aggregation, the parameter server (1) first divides $m$ working machines into $k$ batches, (2) then takes the average of local gradients in each batch, and (3) finally takes the geometric median of those $k$ batch means. With the aggregated gradient, the parameter server performs a gradient descent update.
\begin{algorithm}[htb]
\caption{Byzantine Gradient Descent: Iteration $t\ge 1$}\label{alg:BGD}
\vskip 0.2\baselineskip
\begin{center}
{\bf \color{blue}{\em Parameter server:}}
\end{center}

\begin{algorithmic}[1]
\STATE {\em Initialize:} Let $\theta_0$ be an arbitrary point in $\Theta$; group the $m$ machines into $k$ batches, with the $\ell$-th batch being $\{ (\ell-1) b +1, \ldots, \ell b\}$ for $1\le \ell \le k$.
\STATE Broadcast the current model parameter estimator $\theta_{t-1}$ to all working machines;
\STATE Wait to receive all the gradients reported by the $m$ machines; 
If no message from machine $j$ is received, set $\nabla \tilde{g_j}(\theta_{t-1})$ to be some arbitrary value; 
\STATE  {\em Robust Gradient Aggregation}
\begin{align}
\calA_k (  \vct{g}_t (\theta_{t-1})  ) \gets   \median \left\{  \frac{1}{b} \sum_{j=1}^{b}  g^{(j)}_t   (\theta_{t-1}),  \; \cdots, \;
\frac{1}{b} \sum_{j=n-b+1}^{n} g^{(j)}_t  (\theta_{t-1})\right\}. \label{eq:robust_aggregation}
\end{align}

\STATE Update: $ \theta_{t} ~ \gets ~ \theta_{t-1}   -\eta \times \calA_k \big(\;  \vct{g}_t  (\theta_{t-1}) \; \big)$;
\end{algorithmic}

\vskip \baselineskip
\begin{center}
{\bf \color{blue}{\em Working machine $j$:}}
\end{center}
\begin{algorithmic}[1]
\STATE Compute the gradient $\nabla \bar{f}^{(j)}(\theta_{t-1})$;
\STATE Send $\nabla \bar{f}^{(j)}(\theta_{t-1})$ back to the parameter server;
\end{algorithmic}
\end{algorithm}
Notice that when the number of batches $k=1$, the geometric median of means reduces to the average, \ie,
$$
\calA_1 \{ \vct{g}_t  (\theta_{t-1}) \}  = \frac{1}{m} \sum_{j=1}^m   g_t^{(j)} (\theta_{t-1}).
$$
When $k=m$, the median of means
reduces to the geometric median
$$
\calA_m \{\vct{g}_t  (\theta_{t-1}) \}  = \median \{ g_t^{(1)} (\theta_{t-1}), \ldots, g_t^{(m)} (\theta_{t-1})\}.
$$
Hence, the geometric median of means can be viewed as an interpolation
between the mean and the geometric median. Since the parameter server knows $q$ -- the upper bound on the number of Byzantine machines $q$, it can choose $k$ accordingly. 
We will discuss the choice of $k$ after the statement of our main theorem. 

\subsection{Summary of Convergence Results}
For ease of presentation, we present an informal statement of our main theorem. The precise statement and its proof are given in Section \ref{sec:main_formal}. Our convergence results hold under some technical assumptions on the sample gradients $\nabla f(X_i, \cdot)$, formally stated in Section \ref{sec:main_formal}. Roughly speaking, such assumptions mimic Assumption \ref{ass:pop_risk_smooth} (placed on population risk $F$), and can be viewed a stochastic version of strong convexity and Lipschitz-continuity conditions.
%

\begin{theorem}[Informal] \label{thm:main_informal}
Suppose some mild technical assumptions hold and $2(1+\epsilon)q \le k \le m $ for any arbitrary but fixed constant $\epsilon>0$.
Fix any fixed constant $\alpha \in (\frac{1}{2+2\epsilon},\frac{1}{2} )$ and any $\delta>0$ such that $\delta \le \alpha-q/k$.
There exist universal constants $c_1, c_2>0$ such that if $N/k \ge c_1 C_\alpha^2 \left( d \log (N/k) + \log(1/\delta) \right) $,
then with probability at least $$1-\exp\left(- k D\left( (\alpha -q/k) \| \delta \right) \right),$$ the iterates $\{\theta_t\}$
given by Algorithm \ref{alg:BGD} with $\eta=L/(2M^2)$ satisfy
\begin{align}
\label{informal convergence}
 \| \theta_t - \theta^* \|  & \le \left( \frac{1}{2} + \frac{1}{2} \sqrt{ 1- \frac{L^2}{4M^2} }\right)^t \| \theta_0 - \theta^* \|  + c_2 C_\alpha \sqrt{ \frac{k \left( d+ \log (1/\delta) \right) }{N} }, 
\end{align}
for $t\ge 1$,
where $D(\delta' \| \delta) = \delta' \log \frac{\delta'}{\delta} + (1-\delta') \log \frac{1-\delta'}{1-\delta}$ denotes the binary divergence.
\end{theorem}
The characterization of $c_1$ and $c_2$ can be found in Section \ref{sec:main_formal}. In Theorem \ref{thm:main_informal}, in addition to the non-specified ``technical assumptions'', we also impose assumptions on $\delta$, $\alpha$, $N/k$ and $d$.
Next we illustrate that these conditions can indeed hold simultaneously.

As can be seen later, $\delta$ can be viewed as the expected fraction of batches that are ``statistically bad";  the larger the batch sample size $N/k$ (comparing to $d$), the smaller $\delta$. 
Additionally, up to $q/k$ fraction of the batches may contain Byzantine machines. In total, we may expect $\delta+q/k$ fraction of the batches to be bad.
Theorem \ref{thm:main_informal} says that as long as the total fraction of bad batches is less than $1/2$, we are able to show with high probability, our Byzantine Gradient Descent Method converges exponentially fast.





\begin{remark}
\label{choice k}
In this remark, we discuss the choice of $k$.

When $q=0$,
$k$ can be chosen to be $1$ and
$\log(1/\delta)$ can be chosen to be $d$. Thus the geometric
median of means reduces to simple averaging.
Theorem \ref{thm:main_informal} implies that with probability at least $1-e^{-\Omega(d) }$,
the asymptotic estimation error rate is $\sqrt{d/N}$.

For $q \ge 1$, we can choose $k$ to be $2(1+\epsilon) q$ for an arbitrarily small but fixed
constant $\epsilon>0$ and $\alpha=\frac{2+\epsilon}{4+4\epsilon}$ and $\log(1/\delta)=d$.  In this way, $\alpha-q/k=\frac{\epsilon}{4+4\epsilon}$.
Using the property that $D(\delta'\| \delta) \ge \delta' \log \frac{\delta'}{e \delta}$,
we have that  $D\left( (\alpha -q/k) \| \delta \right)  \ge \Omega(d)$. Hence as long as
$N/k \ge c_1 d \log (N/k)$ for a sufficiently large
universal constant $c_1$, with probability at least $1-e^{-\Omega(qd)}$,
the estimation error of $\theta_t$  converges exponentially
fast to $c_2 \sqrt{dq/N}$ for a universal constant $c_2$.


Based on our analysis, the number of batches $k$ in our Byzantine gradient algorithm provides a trade-off
between the statistical estimation error and the Byzantine failures: With a larger $k$,
our algorithm can tolerate more Byzantine failures, but the estimation error gets larger. However, this trade-off may be an artifact of our proof and may not be fundamental.



\end{remark}

\begin{remark}\label{rmk:gm_approximate}
In this remark, we discuss the practical issues of computing geometric median.

Since exact geometric median may not be computed efficiently in practice, we can  use $(1+\gamma)$-approximate geometric median
in the robust gradient aggregation step \eqref{eq:robust_aggregation}. Moreover, in view of
\prettyref{lmm:geometric_median_robust}, comparing to the exact geometric median,
a $(1+\gamma)$-approximate geometric median induces
an additional deviation term proportional to
$\gamma$ and the maximum norm of the averaged batch gradients. Therefore, in \eqref{eq:robust_aggregation},
we also trim away averaged batch gradients of norm larger than a threshold $\tau$ before computing the
$(1+\gamma)$-approximate geometric median. Since the gradients are of dimension $d$, one can choose
the threshold $\tau=\Theta(d)$ so that with high probability the averaged gradients over Byzantine-free batches
will not be trimmed away. Finally by choosing $\gamma=1/N$, the additional deviation term is ensured
to be $O(d/N)$ and thus it will not affect the final estimation error which is on
the order of $\max\{\sqrt{dq/N},~ \sqrt{d/N}\}.$

\end{remark}
Our algorithm is both computation and communication efficient.
Under the choice of $k$ in Remark \ref{choice k}, the computation and communication cost of our proposed algorithm can be summarized as follows. For estimation error converging to $c_2 \sqrt{dq /N}$, $O(\log N)$ communication rounds are sufficient. In each round, every working machine transmits a $d$-dimensional vector to the parameter server. In terms of computation cost,
in each round, every working machine computes a gradient based on $N/m$ local data samples, which takes $O(Nd/m)$. The parameter server computes the geometric median of means of gradients.
The means of gradients over all batches can be computed in $O(md)$ steps. It is shown in~\cite{cohen2016geometric} that
a $(1+\gamma)$-approximate geometric median can be computed in $O(qd \log^3(1/\gamma) )$ and as we discussed in Remark~\ref{rmk:gm_approximate}, $\gamma=1/N$ suffices for our purpose.
Therefore, in each round, in total the parameter server needs to take $O(md + qd \log^3 (N) )$ steps.

\section{Convergence Results and Analysis}
\label{sec: analysis}
In this section, we present our main results and their proofs.

Recall that in Algorithm \ref{alg:BGD}, the machines are grouped into $k$ batches beforehand. For each batch of machines $ 1 \le \ell \le k$, we define a function $Z_\ell: \Theta \to \reals^d$ to be the {\em difference} between the average of the batch sample gradient functions and the population gradient, i.e.,  $\forall \, \theta\in \Theta$
\begin{align}
\label{batch}
Z_\ell(\theta) &\triangleq \frac{1}{b} \sum_{j= (\ell-1) b +1 }^{ \ell b} \nabla \bar{f}^{(j) } (\theta) - \nabla F(\theta)\nonumber\\
&= \frac{k}{N} \sum_{j= (\ell-1) b +1 }^{ \ell b}  \sum_{i \in \calS_j } \nabla f(X_i, \theta) - \nabla F(\theta),
\end{align}
where the last equality follows from \eqref{local obj} and the fact that batch size $b=m/k$ and local data size $|\calS_j |=N/m$.
Since each function $Z_\ell$ depends on the local data at $\ell$-batch $\{X_i: i \in \calS_j,  (\ell-1) b +1 \le j \le \ell b\}$
and  $X_i$'s are i.i.d., the functions $Z_\ell(\cdot)$'s are also
``independently and identically distributed''. For any given positive precision parameters $\xi_1$ and $\xi_2$ specified later, and $\alpha\in (0, 1/2)$,
define a good event
\begin{align}
\label{good event}
\calE_{\alpha, \xi_1, \xi_2} \triangleq \left\{ \sum_{\ell=1}^k \indc{ \forall \theta: \; C_\alpha  \|Z_\ell (\theta)\| \le \xi_2 \|\theta-\theta^\ast\| + \xi_1} \ge k(1-\alpha) + q \right\}.
\end{align}
Informally speaking, on event $\calE_{\alpha,\xi_1,\xi_2}$, in at least $k(1-\alpha) + q$ batches, the average of the batch sample gradient functions is uniformly close to the population gradient function.

We show our convergence results of Algorithm \ref{alg:BGD} in two steps. The first step is ``deterministic'', showing that our Byzantine
gradient descent algorithm converges exponentially fast on good event $\calE_{\alpha,\xi_1,\xi_2}$. The second
part is ``stochastic'', proving that this good event $\calE_{\alpha, \xi_1, \xi_2}$ happens with high probability.

\subsection{Convergence of Byzantine Gradient Descent on $\calE_{\alpha, \xi_1, \xi_2}$}
%

We consider a fixed execution.
Recall that $\calB_t$ denotes the set of Byzantine machines at iteration $t$ of the given execution,
which could change across iterations. Define a vector of functions $\vct{g}_t(\cdot)$ with respect to $\calB_t$ as:
$$\vct{g}_t(\theta) = (g^{(1)}_t (\theta), \ldots, g^{(m)}_t (\theta) ), ~ \forall \, \theta
$$
such that $ \forall \, \theta $,
$$
g^{(j)}_t (\theta) = \begin{cases}
\nabla \bar{f}^{(j)}(\theta) & \text{ if } j \notin \calB_t \\
\star &  \text{ o.w. },
\end{cases}
$$
where $\star$ is arbitrary\footnote{By ``arbitrary'' we mean that  $g^{(j)}_t (\cdot)$ cannot be specified.}. That is,  $g^{(j)}_t (\cdot)$ is the true gradient function
$\bar{f}^{(j)}(\cdot)$  if machine $j$ is not Byzantine at iteration $t$, and arbitrary otherwise. It is easy to see that the definition of $\vct{g}_t(\cdot)$ is consistent with the definition of $\vct{g}_t (\theta_{t-1})$ in \eqref{received gradients}.
Define $\tilde{Z}_\ell (\cdot)$ for each $\theta$ as
\begin{align}
\label{received batch}
\tilde{Z}_\ell (\theta) \triangleq  \frac{1}{b} \sum_{j=(\ell-1)b +1}^{\ell b } g^{(j)}_t  (\theta) - \nabla F(\theta).
\end{align}
By definition of $g^{(j)}_t (\cdot)$, for any $\ell$-th batch such that
$$\{ b(\ell-1) +1, \ldots, b \ell \} \cap \calB_t = \emptyset,$$  \ie, it does not contain
any Byzantine machine at iteration $t$, it holds that
$\tilde{Z}_\ell (\theta) = Z_\ell (\theta)$ for  all  $\theta$, where $ Z_\ell (\cdot)$ is defined in \eqref{batch}.
%
%
%
%
\begin{lemma}\label{lmm:aggregated_gradient}
On event $\calE_{\alpha, \xi_1, \xi_2}$,  for every iteration $ t \ge 1$, we have
$$
  \left \| \calA_k \left( \vct{g}_t (\theta) \right)  - \nabla F(\theta) \right\|
\le  \xi_2 \|\theta-\theta^\ast\| + \xi_1, \quad  \forall \theta \in \Theta.
$$
\end{lemma}
\begin{proof}
By definition of $\calA_k$ in \eqref{eq:robust_aggregation},
for any fixed $\theta$,
$$
\calA_k (  \vct{g}_t (\theta)  ) =   \median \left\{  \frac{1}{b} \sum_{j=1}^{b}  g^{(j)}_t (\theta) , \; \cdots, \;
\frac{1}{b} \sum_{j=m-b+1}^{m} g^{(j)}_t  (\theta) \right\}
$$
Since geometric median is invariant with translation, it follows that
$$
\calA_k (  \vct{g}_t (\theta)  ) - \nabla F(\theta) =\median \left\{ \tilde{Z}_1 (\theta),   \; \cdots, \;
\tilde{Z}_m (\theta) \right\} .
$$
On event $\calE_{\alpha, \xi_1, \xi_2}$, at least $k(1-\alpha)+q$ of the $k$ batches $\{Z_\ell: 1 \le \ell \le k\}$
satisfy $C_\alpha \|Z_\ell (\theta)\|  \le \xi_2 \| \theta-\theta^* \| + \xi_1$ uniformly. Moreover,
for Byzantine-free batch $\ell$, it holds that $\tilde{Z}_\ell (\cdot) = Z_\ell (\cdot)$.
Hence, at least $k(1-\alpha)$ of the $k$ {\em received} batches $\{ \tilde{Z}_\ell: 1 \le \ell \le k\}$ satisfy
$C_\alpha \|\tilde{Z}_\ell (\theta)\| \le \xi_2 \| \theta-\theta^* \| + \xi_1$ uniformly.
The conclusion readily  follows from Lemma \ref{lmm:geometric_median_robust} with $\gamma=0$.



\end{proof}

\subsubsection{\bf Convergence of Approximate Gradient Descent}

Next, we show a convergence result of an approximate gradient descent,
which might be of independent interest.
For any $\theta \in \Theta$, define a new $\theta'$ as
\begin{align}
\theta'= \theta -\eta \times \nabla F(\theta). \label{eq:pop_gd}
\end{align}
We remark that the above update is one step of population gradient descent
given in \eqref{eq:pop_gd o}.

\begin{lemma}\label{lmm:convg_pop_gd}
Suppose Assumption \ref{ass:pop_risk_smooth} holds. If we choose
the step size $\eta=L/(2M^2)$, then  ${\theta'}$ defined in \eqref{eq:pop_gd} satisfies that
\begin{align}
\| \theta' - \theta^* \| \le  \sqrt{1- L^2/(4M^2)} \; \| \theta - \theta^*\|. \label{eq:pop_convergence}
\end{align}
\end{lemma}
The proof of Lemma \ref{lmm:convg_pop_gd} is rather standard, and is presented in Section \ref{pf app convergence} for completeness.
Suppose
that for each $t \ge 1$, we have access to gradient function $G_t(\cdot)$, which
satisfy the uniform deviation bound:
\begin{equation}
\left\Vert G_t(\theta)-\nabla F(\theta)\right\Vert \le \xi_1 +  \xi_2 \| \theta-\theta^*\|,
\quad\forall\theta,  \label{eq:deviation}
\end{equation}
for two positive precision parameters $\xi_1, \xi_2$ that are {\em independent} of $t$.
Then we perform the following approximate gradient descent as a surrogate for population gradient descent:
\begin{align}
\theta_{t}=\theta_{t-1}-\eta \times G_t(\theta_{t-1}). \label{eq:noisy_gd}
\end{align}
The following lemma establishes the convergence of the approximate gradient descent.
\begin{lemma}\label{lmm:noisy_grad}
Suppose Assumption \ref{ass:pop_risk_smooth} holds, and choose $\eta=L/(2M^2)$.
If \eqref{eq:deviation} holds for each $t\ge1$ and
$$\rho \triangleq 1-\sqrt{1-L^2/(4M^2)} - \xi_2L/(2M^2) >0,$$
then the iterates $\{\theta_{t}\}$ in \eqref{eq:noisy_gd} satisfy
\[
\left\Vert \theta_{t}-\theta^{*}\right\Vert\le \left( 1-\rho  \right)^{t}\left\Vert \theta_{0}-\theta^{*}\right\Vert + \eta \xi_1/\rho.
\]
\end{lemma}
\begin{remark}
As it can be seen later, the precision parameter $\xi_2$ can be chosen to be a function of $N/k$ such that $\xi_2\to 0$ as $N/k\diverge$. Thus, there exists $\xi_2$ for $\rho$ defined in Lemma \ref{lmm:noisy_grad} to be positive.
\end{remark}

\begin{proof}[Proof of Lemma \ref{lmm:noisy_grad}]
Fix any $t\ge 1$, we have 
\begin{align*}
 \left\Vert \theta_{t}-\theta^{*}\right\Vert
= & \left\Vert \theta_{t-1}-\eta G_t(\theta_{t-1})-\theta^{*}\right\Vert \\
= & \left\Vert \theta_{t-1}-\eta\nabla F(\theta_{t-1})-\theta^{*}+\eta\left(\nabla F(\theta_{t-1})-G_t(\theta_{t-1})\right)\right\Vert\\
\le &  \left\Vert \theta_{t-1}-\eta\nabla F(\theta_{t-1})-\theta^{*}\right\Vert + \eta \left\Vert \nabla F(\theta_{t-1})-G_t(\theta_{t-1})\right\Vert.
\end{align*}
It follows from Lemma \ref{lmm:convg_pop_gd} that
$$
\left\Vert \theta_{t-1}-\eta\nabla F(\theta_{t-1})-\theta^{*}\right\Vert \le \sqrt{1- L^2/(4M^2)} \left\Vert \theta_{t-1}-\theta^{*}\right\Vert
$$
and from \eqref{eq:deviation} that
$$
\left\Vert \nabla F(\theta_{t-1})-G_t(\theta_{t-1})\right\Vert \le \xi_1 + \xi_2 \| \theta_{t-1} - \theta^* \| .
$$
Hence, $$
\left\Vert \theta_{t}-\theta^{*}\right\Vert
\le \left(\sqrt{1- L^2/(4M^2)} + \eta \xi_2 \right)  \| \theta_{t-1} - \theta^* \|_2 + \eta \xi_1 .
$$
A standard telescoping argument then yields that
\begin{align*}
\left\Vert \theta_{t}-\theta^{*}\right\Vert  &\le  \left( 1-\rho \right)^{t}\left\Vert \theta_{0}-\theta^{*}\right\Vert+ \eta \xi_1
 \sum_{\tau=0}^{t-1} \left( 1-\rho \right)^{\tau}\\
&\le\left( 1-\rho  \right)^{t}\left\Vert \theta_{0}-\theta^{*}\right\Vert +
 \eta \xi_1/\rho,
 \end{align*}
where $ \rho =  1-\sqrt{1-L^2/(4M^2)} - \xi_2L/(2M^2)$ and $\eta= L/(2M^2)$.
\end{proof}

\subsubsection{\bf Convergence of Byzantine Gradient Descent  on Good Event $\calE_{\alpha, \xi_1, \xi_2}$}

With Lemma \ref{lmm:aggregated_gradient} and the convergence of the approximate gradient descent (Lemma \ref{lmm:noisy_grad}), we show that Algorithm \ref{alg:BGD} converges exponentially fast on good event $\calE_{\alpha, \xi_1, \xi_2}$.
\begin{theorem}\label{thm:convg_BGD}
Suppose event $\calE_{\alpha, \xi_1, \xi_2}$ holds and iterates $\{\theta_t\}$ are given by Algorithm \ref{alg:BGD}
with $\eta=L/(2M^2)$. If $\rho= 1-\sqrt{1-L^2/(4M^2)} - \xi_2 L/(2M^2)>0$ as defined in Lemma \ref{lmm:noisy_grad}, then
\begin{align}
\| \theta_t - \theta^* \| \le (1-\rho)^t \| \theta_0 - \theta^* \| + \eta \xi_1/\rho.
\end{align}
\end{theorem}
\begin{proof}
In Algorithm \ref{alg:BGD}, at iteration $t$, the parameter server updates the model parameter $\theta_{t-1}$ using the approximate gradient $\calA_k \left( \vct{g}_t (\theta_{t-1})\right)$ -- the value of the approximate gradient function $\calA_k \left( \vct{g}_t (\cdot) \right)$ evaluated at $\theta_{t-1}$. From Lemma \ref{lmm:aggregated_gradient}, we know that on event $\calE_{\alpha, \xi_1, \xi_2}$
$$
  \left \| \calA_k \left( \vct{g}_t (\theta) \right)  - \nabla F(\theta) \right\|
\le  \xi_2 \|\theta-\theta^\ast\| + \xi_1, \quad  \forall \theta \in \Theta.
$$
The conclusion then follows from
Lemma \ref{lmm:noisy_grad} by setting $G_t (\theta)$ to be $ \calA_k \left( \vct{g}_t (\theta) \right)$.
\end{proof}

\subsection{Bound Probability of Good Event $\calE_{\alpha, \xi_1, \xi_2}$}
Recall that for each batch $\ell$ for $1\le \ell \le k$, $Z_{\ell}$ is defined in \eqref{batch} w.r.t.\ the data samples collectively kept by the machines in this batch. Thus, function $Z_{\ell}$ is random.
The following lemma gives a lower bound to the probability of good event $\calE_{\alpha, \xi_1, \xi_2}$.
%
%
\begin{lemma}\label{lmm:geometric_median}
Suppose for all $1 \le \ell \le k$, $Z_{\ell}$ satisfies   
\begin{align}
\prob{ \forall \theta: C_\alpha \|Z_{\ell} (\theta)  \|  \le \xi_2 \| \theta-\theta^\ast\|  + \xi_1} \ge 1- \delta \label{eq:assump_geometric}
\end{align}
for any $\alpha \in (q/k,1/2)$ and $0<\delta \le \alpha - q/ k$.
Then
\begin{align}
\prob{ \calE_{\alpha, \xi_1, \xi_2}} \ge 1-   e^{ - k D( \alpha - q/ k  \| \delta ) }. \label{eq:conclusion_geometric}
\end{align}
\end{lemma}
\begin{proof}
Let $T \sim \Binom(k,1-\delta)$.
By assumption \eqref{eq:assump_geometric},
$$
 \sum_{\ell=1}^{k} \indc{ \forall \theta: C_\alpha \|Z_{\ell} (\theta)  \|_2 \le \xi_2 \| \theta-\theta^\ast\|  + \xi_1}$$  first-order stochastically dominates $T$, i.e.,
\begin{align}
\label{batch grouping 1}
&\prob{   \sum_{\ell=1}^{k} \indc{ \forall \theta:  C_\alpha\|Z_{\ell} (\theta)  \|_2 \le \xi_2 \| \theta-\theta^\ast\|  + \xi_1}  \ge  k (1-\alpha) +q } \ge  \prob{ T \ge  k (1-\alpha) +q }.
\end{align}
By Chernoff's bound for binomial distributions, the following holds:
\begin{align}
\label{batch grouping 2}
\prob{ T \ge  k (1-\alpha) +q } \ge 1-  e^{ -k D(  \alpha - q/k \| \delta ) }.
\end{align}
Combining \eqref{batch grouping 1} and \eqref{batch grouping 2} together, we conclude \eqref{eq:conclusion_geometric}.

\end{proof}

It remains to show the uniform convergence of $Z_\ell$ as required by \eqref{eq:assump_geometric}. To this end, we need to impose a few technical assumptions that are rather standard \cite{vershynin2010nonasym}.
Recall that gradient $\nabla f(X, \theta)$ is random as the input $X$ is random.
We assume gradient $\nabla f(X, \theta^*)$ is sub-exponential. The definition and some related concentration properties of sub-exponential random variables are presented in Section \ref{subexpon prelim} for completeness.
\begin{assumption}\label{ass:gradient_sub_gaussian}
There exist positive constants $\sigma_1$ and $\alpha_1$ such that for any
unit vector $v \in B$, $\Iprod{\nabla f(X, \theta^*) }{v}$
is sub-exponential with scaling parameters $\sigma_1$ and $\alpha_1$, \ie,
$$
\sup_{v \in B}  \expect{ \exp \left( \lambda \Iprod{ \nabla f(X, \theta^*) }{v} \right) } \le e^{\sigma_1^2 \lambda^2/2},
\quad \forall |\lambda| \le \frac{1}{\alpha_1},
$$
where $B$ denotes the unit sphere $\{\theta: \| \theta\|_2=1\}$.
\end{assumption}

Intuitively speaking, Assumption \ref{ass:gradient_sub_gaussian} is placed to ensure that, with high probability, using the {\em true} sample gradient for individual batches, we are able to ``identify" the optimal model $\theta^*$. That is, $(1/n)\sum_{i=1}^n \nabla f(X_i, \theta^*)$ concentrates around its mean $\nabla F(\theta^*) =0$.
%
\begin{lemma}\label{lmm:gradient_sub_gaussian}
Suppose Assumption \ref{ass:gradient_sub_gaussian} holds.
For any $\delta \in (0,1)$ and any positive integer $n$, let 
\begin{align}
\Delta_1 (n, d, \delta, \sigma_1)=\sqrt{2} \sigma_1 \sqrt{ \frac{d \log 6 + \log(3/\delta) }{ n}  }.\label{eq:s}
\end{align}
If  $\Delta_1(n, d, \delta, \sigma_1) \le \sigma_1^2/\alpha_1$, then
$$
 \prob{ \left \|  \frac{1}{n} \sum_{i=1}^n  \nabla f(X_i, \theta^* )  - \nabla F(\theta^* ) \right \|
 \ge 2\Delta_1(n, d, \delta, \sigma_1)} \le \frac{\delta}{3}.
$$
\end{lemma}
\begin{remark}
By definition of $\Delta_1 (n, d, \delta, \sigma_1)$,
$\Delta_1 (n, d, \delta, \sigma_1)$ is a non-increasing function of $n$. In particular,
 for fixed $\delta$ and $\sigma_1$, if $d=o(n)$,
\begin{align*}
\Delta_1 (n, d, \delta, \sigma_1)=\sqrt{2} \sigma_1 \sqrt{ \frac{d \log 6 + \log(3/\delta) }{ n}  } \to 0 ~~~~~ \text{as} ~~ n \diverge.
\end{align*}
Thus, if in addition $\alpha_1$ is assumed to be fixed, then for sufficiently large $n$,
$\Delta_1(n, d, \delta, \sigma_1) \le \sigma_1^2/\alpha_1$ holds.
\end{remark}
With a little abuse of notation, we write $\Delta_1 (n, d, \delta, \sigma_1)$ as $\Delta_1$ or $\Delta_1 (n)$ for short when its meaning is clear from the context.
Also, we let $ \nabla \bar{f}_n (\theta)$ denote $\frac{1}{n} \sum_{i=1}^n  \nabla f(X_i, \theta )$ for ease of exposition.

\begin{proof}[Proof of Lemma \ref{lmm:gradient_sub_gaussian}]
Let $\calV=\{v_1, \ldots, v_{N_{1/2} }\}$ denote an $\frac{1}{2}$-cover of unit sphere $B$.
It is shown in \cite[Lemma 5.2, Lemma 5.3]{vershynin2010nonasym} that $\log N_{1/2} \le d \log 6$, and
$$
\left \| \nabla \bar{f}_n (\theta^\ast)  - \nabla F(\theta^* ) \right \|
\le 2 \sup_{v \in \calV}  \left\{ \iprod{ \nabla \bar{f}_n (\theta^\ast)     - \nabla F(\theta^* ) }{v} \right\}.
$$
Note that since $\nabla F(\theta^*)=0$, it holds that $\nabla f(X_i, \theta^* )  - \nabla F(\theta^* )=\nabla f(X_i, \theta^* )$. By Assumption \ref{ass:gradient_sub_gaussian} and the condition that $\Delta_1 \le \sigma_1^2/\alpha_1$,
it follows from concentration inequalities for sub-exponential random variables given in Theorem \ref{thm:sum_sub_exponential} that, for $v\in \calV$
$$
\prob{ \iprod{ \nabla \bar{f}_n (\theta^\ast)   - \nabla F(\theta^* ) }{v} \ge \Delta_1}
\le \exp \left( -n \Delta_1^2/(2\sigma_1^2) \right).
$$
Recall that in $\calV$ contains at most $6^d$ vectors. In view of the union bound, it further yields that
\begin{align*}
\prob{ 2 \sup_{v \in \calV}  \left\{ \iprod{\nabla \bar{f}_n (\theta^\ast)   - \nabla F(\theta^* ) }{v} \right\} \ge 2 \Delta_1}
&\le 6^d \exp \left( -n \Delta_1^2/(2\sigma_1^2) \right) \\
&= \exp \left( -n \Delta_1^2/(2\sigma_1^2) + d \log 6 \right).
\end{align*}
Therefore,
\begin{align*}
\prob{ \left \| \nabla \bar{f}_n (\theta^\ast)   - \nabla F(\theta^* ) \right \| \ge 2 \Delta_1}
\le \exp \left( -n \Delta_1^2/(2\sigma_1^2) + d \log 6 \right).
\end{align*}
%
%
\end{proof}

In addition to the ``identifiability" of the optimal $\theta^*$ using sample gradients $\nabla f(X, \theta^* )$, similar to the smoothness requirements of the population gradient $\nabla F (\cdot)$ stated in Assumption \ref{ass:pop_risk_smooth}, some smoothness properties (in stochastic sense) of the sample gradients $\nabla f(X, \cdot )$ are also desired. Next, we define gradient difference
\begin{align}
\label{difference1}
h(x,\theta) \triangleq \nabla f(x,\theta)-\nabla f(x,\theta^{*}),
\end{align}
which characterizes the deviation of random gradient $\nabla f(x, \theta)$
from $\nabla f(x, \theta^*)$.  Note that
\begin{align}
\label{difference2}
\expect{h(X,\theta)}=\nabla F(\theta)-\nabla F(\theta^{*})
\end{align}
for each $\theta$. The following assumptions ensure that
for every $\theta$,
$h(x,\theta)$ normalized by $\|\theta-\theta^*\|$ is
also sub-exponential.

\begin{assumption}\label{ass:gradient_sub_exp}
There exist positive constants $\sigma_2$ and $ \alpha_2$ such that
for any $\theta \in \Theta$ with $\theta \neq \theta^*$ and unit vector $v \in B$,
$\Iprod{ h(X,\theta)-\expect{h(X,\theta)} }{v}/\|\theta-\theta^*\|$
is sub-exponential with scaling parameters $(\sigma_2, \alpha_2)$, \ie,
for all $| \lambda| \le \frac{1}{\alpha_2}$,
$$
\sup_{\theta \in \Theta, v \in B} \expect{\exp \left( \frac{  \lambda \Iprod{ h(X,\theta)-\expect{h(X,\theta) } }{v} }{\|\theta-\theta^*\|} \right)}
\le e^{\sigma_2^2 \lambda^2 /2}.
$$
\end{assumption}

The following lemma bounds the deviation of $(1/n) \sum_{i=1}^n  h(X_i, \theta)$
from $\expect{h(X,\theta)}$ for every $\theta \in \Theta$ under Assumption \ref{ass:gradient_sub_exp}.
Its proof is similar to that of Lemma \ref{lmm:gradient_sub_gaussian} and thus is omitted.
\begin{lemma}\label{lmm:gradient_sub_exp}
Suppose Assumption \ref{ass:gradient_sub_exp} holds and fix any $\theta \in \Theta$.
Let
\begin{align}
\Delta'_1 (n, d, \delta, \sigma_2)=\sqrt{2} \sigma_2  \sqrt{ \frac{d \log 6 + \log(3/\delta) }{ n}  }. \label{eq:s_prime}
\end{align}
If  $\Delta_1'(n, d, \delta, \sigma_2) \le \sigma_2^2 /\alpha_2$, then
$$
\prob{ \left \|\frac{1}{n}\sum_{i=1}^{n}h (X_{i}, \theta)-\expect{h (X, \theta) } \right \|
>2  \Delta'_1(n, d, \delta, \sigma_2) \| \theta-\theta^*\| } \le \frac{\delta}{3}.
$$
\end{lemma}
\begin{remark}
Similar to $\Delta_1(n, d, \delta, \sigma_2)$,  if $\delta$, $\sigma_1$, and $\sigma_2$ are fixed,
and $d=o(n)$, then for all sufficiently large $n$, it holds that
$\Delta_1'(n, d, \delta, \sigma_2) \le \sigma_2^2 /\alpha_2.$
\end{remark}
We write $\Delta_1'(n, d, \delta, \sigma_2)$ as $\Delta_1'$ or $\Delta_1'(n)$ for short.


Assumption \ref{ass:gradient_sub_gaussian} and Assumption \ref{ass:gradient_sub_exp} can be potentially relaxed at an expense of looser concentration bounds.
Note that Assumption \ref{ass:gradient_sub_exp}, roughly speaking, only imposes some smoothness condition w.\ r.\ t.\ the optimal model $\theta^*$. To mimic the Lipschitz continuity of the sample gradients (in stochastic sense), we impose the following assumption, which holds automatically if we strengthen Assumption \ref{ass:gradient_sub_exp} by replacing $\theta^*$ with an arbitrary $\theta'$ such that $\theta\not=\theta'$.
%
%
In general, Assumption \ref{ass:bounded_hessian} is strictly weaker than the strengthened version of Assumption \ref{ass:gradient_sub_exp}.
\begin{assumption}\label{ass:bounded_hessian}
For any $\delta \in (0,1)$, there exists an $M'=M'(n, \delta)$ that is non-increasing in $n$ such that
$$
\prob{ \sup_{ \theta, \theta' \in \Theta: \theta \neq \theta'}  \frac{ \| \frac{1}{n} \sum_{i=1}^n  \pth{\nabla  f(X_i,\theta) - \nabla f(X_i, \theta')} \| }{\| \theta - \theta'\| }  \le M'} \ge 1-\frac{\delta}{3}.
$$
\end{assumption}



 With Assumption \ref{ass:gradient_sub_gaussian}--Assumption \ref{ass:bounded_hessian}, we apply the celebrated $\epsilon$-net argument to prove the averaged random gradients {\em uniformly} converges to $\nabla F(\cdot)$.

 For a given real number $r>0$, define $\Delta_2$ as follows.
\begin{align}
\Delta_2 (n) = \sigma_2 \sqrt{\frac{2}{n} } \sqrt{d \log \frac{18 M \vee M'}{ \sigma_2} + \frac{1}{2} d \log \frac{n}{d} + \log \left( \frac{6\sigma_2^2 r\sqrt{n} }{\alpha_2 \sigma_1 \delta} \right)}. \label{eq:def_t}
\end{align}

\begin{proposition}\label{prop:uniform_converg}
Suppose Assumption \ref{ass:gradient_sub_gaussian} -- Assumption \ref{ass:bounded_hessian} hold,
and  $\Theta \subset  \{ \theta: \| \theta-\theta^*\| \le r \sqrt{d} \}$ for some positive parameter $r$.
For any $\delta \in (0,1)$ and any integer $n$, recall $\Delta_1$ defined in \eqref{eq:s} and
define $\Delta_2$ as in \eqref{eq:def_t}.
If  $\Delta_1 \le \sigma_1^2/\alpha_1$ and $\Delta_2 \le \sigma_2^2/\alpha_2$,
then
$$
\prob{ \forall \theta \in \Theta: \left \|  \frac{1}{n} \sum_{i=1}^n \nabla f(X_i, \theta) - \nabla F(\theta) \right\|
\le 8 \Delta_2 \| \theta-\theta^*\|+ 4 \Delta_1} \ge 1-\delta.
$$
\end{proposition}

\begin{proof}
The proof is based on the classical $\epsilon$-net argument.
Let
$$
\tau= \frac{\alpha_2 \sigma_1  }{2 \sigma_2^2} \sqrt{\frac{d}{n}}
\quad \text{ and }  \quad \ell^\ast = \lceil r\sqrt{d}/\tau \rceil.
$$
Henceforth, for ease of exposition, we assume $\ell^*$ is an integer. For integers $1 \le \ell \le \ell^*$, define
$$
\Theta_\ell \triangleq \left\{ \theta: \| \theta- \theta^* \| \le \tau \ell \right\}.
$$
For a given $\ell$, let $\theta_1, \ldots, \theta_{N_{\epsilon_{\ell}}}$ be an $\epsilon_{\ell}$-cover of
$\Theta_\ell$, where  $\epsilon_{\ell}$ is given by
$$
\epsilon_{\ell}= \frac{\sigma_2 \tau \ell}{M \vee M'}\sqrt{ \frac{d}{n} },
$$
where $M \vee M'= \max \{M, M'\}$.
By \cite[Lemma 5.2]{vershynin2010nonasym},
$\log N_{\epsilon_{\ell}} \le d \log (3 \tau \ell/\epsilon_{\ell})$.
Fix any $\theta \in \Theta_\ell$. There exists a $1 \le j_{\ell} \le N_{\epsilon_{\ell}}$
such that $\|\theta - \theta_{ j_{\ell}} \|_2 \le \epsilon_{\ell}$.
Recall that we let $ \nabla \bar{f}_n (\theta)$ denote $\frac{1}{n} \sum_{i=1}^n  \nabla f(X_i, \theta )$.
By triangle's inequality,
\begin{align}
&\left \| \nabla \bar{f}_n (\theta) - \nabla F(\theta) \right\| \le  \left \| \nabla F(\theta)  - \nabla F(\theta_{j_{\ell}}) \right\| +\left \|  \nabla \bar{f}_n (\theta) - \nabla \bar{f}_n (\theta_{j_{\ell}} )  \right\| + \left \|  \nabla \bar{f}_n (\theta_{j_{\ell}} )  - \nabla F(\theta_{j_{\ell}}) \right\|.
\label{eq:triangle_1}
\end{align}
In view of Assumption \ref{ass:pop_risk_smooth},
\begin{align}
\left \| \nabla F(\theta)  - \nabla F(\theta_{j_{\ell}}) \right\|
\le M \| \theta- \theta_{j_{\ell}} \| \le M \epsilon_{\ell}, \label{eq:approx_error_1}
\end{align}
where the last inequality holds because by the construction of $\epsilon$-net, and the fact that for a given $\theta$, $\theta_{j_{\ell}}$ is chosen in such a way that $\| \theta- \theta_{j_{\ell}} \| \le \epsilon_{\ell}$.

Define event
$$
\calE_1 = \left\{ \sup_{ \theta, \theta' \in \Theta: \theta \neq \theta'} \frac{ \|  \nabla \bar{f}_n (\theta)   -  \nabla \bar{f}_n (\theta')  \| }{\| \theta - \theta'\| }  \le M'  \right \}.
$$
By Assumption \ref{ass:bounded_hessian}, we have $\prob{\calE_1} \ge 1-\delta/3$.  On event $\calE_1$, we have
\begin{align}
\sup_{\theta \in \Theta}  \left \| \nabla \bar{f}_n (\theta)   -  \nabla \bar{f}_n ( \theta_{j_{\ell}} )   \right\|
& \le  M' \epsilon_{\ell}.\label{eq:approx_error_2}
\end{align}
By triangle's inequality again,
\begin{align}
 \left \|   \nabla \bar{f}_n ( \theta_{j_{\ell}} )  - \nabla F(\theta_{j_{\ell}} )   \right\|
 &\le  \left \|  \nabla \bar{f}_n ( \theta^* )  - \nabla F(\theta^* ) \right \| + \left \|  \nabla \bar{f}_n ( \theta_{j_{\ell}} )  -   \nabla \bar{f}_n ( \theta^*)  -\pth{\nabla F(\theta_{j_{\ell}}) - \nabla F(\theta^*)} \right \| \nonumber\\
& \le  \left \| \nabla \bar{f}_n ( \theta^* )  - \nabla F(\theta^* ) \right \| + \left \|  \frac{1}{n} \sum_{i=1}^n    h(X_i,\theta_{j_{\ell}})  -  \expect{h(X,\theta_{j_{\ell}})}  \right \|, \label{eq:triangle_2}
\end{align}
where function $h(x, \cdot)$ is defined in \eqref{difference1}.
Define event
$$
\calE_2= \left\{ \left \| \nabla \bar{f}_n ( \theta^* )  - \nabla F(\theta^* ) \right \|
 \le 2\Delta_1 \right\}
$$
and event
$$
\calF_\ell = \left\{
 \sup_{ 1 \le j \le N_\epsilon} \left \|\frac{1}{n}\sum_{i=1}^{n}h(X_i,\theta_j)- \expect{h(X,\theta_j)}\right \|
\le 2 \tau \ell \Delta_2  \right\},
$$
where $\Delta_2$ is defined in \eqref{eq:def_t} and satisfies
\begin{align}
\label{eq:def_t1}
\Delta_2 = \sqrt{2} \sigma_2 \sqrt{ \frac{d \log 6 + d \log (3\tau  \ell/\epsilon_{\ell}) + \log(3 \ell^*/\delta)  }{ n}  } .
\end{align}
In \eqref{eq:def_t}, note that $\Delta_2$ is independent of $\ell$, due to the choice of $\epsilon_{\ell}$ made earlier. It is easy to check that \eqref{eq:def_t} and \eqref{eq:def_t1} are equivalent.

Since $\Delta_1 \le \sigma_1^2/\alpha_1$, it follows from Lemma \ref{lmm:gradient_sub_gaussian} that
$ \prob{ \calE_2} \ge 1- \delta/3. $ Similarly, since $\Delta_2 \le \sigma_2^2 /\alpha_2$,
by Lemma \ref{lmm:gradient_sub_exp},
$\prob{\calF_\ell } \ge 1-\delta/(3\ell^*) $. In particular,
\begin{align}
\label{2b0}
\nonumber
\prob{\calF_{\ell}^c}& = \prob{ \sup_{ 1 \le j \le N_{\epsilon_{\ell}}} \left \|\frac{1}{n}\sum_{i=1}^{n}h(X_i,\theta_j)- \expect{h(X,\theta_j)}\right \|
> 2 \tau \ell \Delta_2 }\\
\nonumber
&\le \sum_{j=1}^{N_{\epsilon_{\ell}}} \prob{ \left \|\frac{1}{n}\sum_{i=1}^{n}h(X_i,\theta_j)- \expect{h(X,\theta_j)}\right \|
> 2 \tau \ell \Delta_2 }\\
& \le \frac{\delta}{3\ell^*} \frac{1}{\pth{\frac{3\tau\ell}{\epsilon_{\ell}}}^d} \pth{\frac{3\tau\ell}{\epsilon_{\ell}}}^d = \frac{\delta}{3\ell^*}.
\end{align}
Therefore, we have $\prob{\calF_\ell } \ge 1-\delta/(3\ell^*) $.

In conclusion, by combining \eqref{eq:triangle_1}, \eqref{eq:approx_error_1}, \eqref{eq:approx_error_2} and \eqref{eq:triangle_2},
it follows that on event $\calE_1 \cap \calE_2 \cap \calF_\ell $,
\begin{align*}
\sup_{\theta \in \Theta_\ell} \left \| \nabla \bar{f}_n ( \theta ) - \nabla F(\theta) \right\|
&\le (M+M') \epsilon_{\ell} + 2 \Delta_1 + 2 \Delta_2 \tau \ell \\
&\le 4 \Delta_2 \tau \ell+2\Delta_1,
\end{align*}
where the last inequality holds due to $(M \vee M')\epsilon_{\ell} \le \Delta_2 \tau \ell$.
Let
$$
\calE=\calE_1 \cap \calE_2   \cap \left( \cap_{\ell=1}^{ \ell^* } \calF_\ell \right).
$$
It follows from the union bound, $\prob{\calE} \ge 1- \delta.$
Moreover, suppose event $\calE$ holds. Then for all $\theta \in \Theta_{\ell^*}$,
there exists an $1 \le \ell \le \ell^*$ such that $(\ell-1)\tau <\| \theta - \theta^* \| \le \ell \tau$.
If $\ell \ge 2$, then $\ell \le 2(\ell-1)$ and thus
$$
 \left \| \nabla \bar{f}_n ( \theta )- \nabla F(\theta) \right\| \le 4 \Delta_2 \tau \ell+2\Delta_1
  \le 8 \Delta_2 \| \theta-\theta^* \|+ 2 \Delta_1.
$$
If $\ell =1$, then
$$
 \left \| \nabla \bar{f}_n ( \theta )- \nabla F(\theta) \right\| \le 4 \Delta_2 \tau +2\Delta_1
 \le  4 \Delta_1,
$$
where the last inequality follows from our choice of $\tau$ and the
assumption that $\Delta_2 \le \sigma_2^2/\alpha_2$ and $\Delta_1 \ge \sigma_1 \sqrt{d/n}$.
In conclusion, on event $\calE$,
$$
\sup_{\theta \in \Theta_{\ell^*} } \left \| \nabla \bar{f}_n ( \theta ) - \nabla F(\theta) \right\|
\le 4 \Delta_1 +  8 \Delta_2  \| \theta-\theta^*\| .
$$
The proposition follows by the assumption that $\Theta \subset \Theta_{\ell^*}$.



\end{proof}

\begin{theorem}\label{thm:unif_convg}
Suppose Assumption \ref{ass:gradient_sub_gaussian} -- Assumption \ref{ass:bounded_hessian} hold,
and  $\Theta \subset  \{ \theta: \| \theta-\theta^*\| \le r \sqrt{d} \}$ for some positive parameter $r$.
For any $\delta \in (0,1)$ and any integer $n$, define $\Delta_1 (n)$ and $\Delta_2 (n)$ as in \eqref{eq:s} and \eqref{eq:def_t}, respectively.
%
If $\Delta_1 (N/k) \le \sigma_1^2/\alpha_1$ and $\Delta_2 (N/k) \le \sigma_2^2/\alpha_2$, then for every $1 \le \ell \le k$,
$$
\prob{ \forall \theta \in \Theta: C_\alpha \| Z_\ell (\theta) \| \le \xi_2 \| \theta-\theta^*\|+ \xi_1} \ge 1-\delta,
$$
where $
\xi_1= 4 C_\alpha \times \Delta_1 (N/k)$
and $\xi_2 = 8 C_\alpha \times \Delta_2 (N/k).$
\end{theorem}

\begin{proof}
Recall that $Z_{\ell}$ is defined in \eqref{batch}. Note that for each $1 \le \ell \le k$, $Z_\ell$ has the same distribution as
the average of $N/k$ i.i.d.\ random gradients $f(X_i, \theta)$ subtracted by
$\nabla F(\theta)$. Hence, Theorem \ref{thm:unif_convg} readily follows from Proposition \ref{prop:uniform_converg}.
\end{proof}
\begin{remark}
Suppose $\sigma_1,$ $\alpha_1,$ $\sigma_2,$ $\alpha_2$ are all of $\Theta(1)$,
$\log (M \vee M') = O(\log d)$, $\log (1/\delta) = O(d)$ and $\log r = O(d \log (N/k) )$.
In this case, Theorem \ref{thm:unif_convg} implies that if $N/k =\Omega(C_\alpha^2 d \log (N/k))$, then
$$
\xi_1= O\left(C_\alpha \sqrt{kd /N } \right) ~~\text{and}~~
\xi_2 =O\left(C_\alpha \sqrt{ kd \log (N/k)/N } \right).
$$   In particular, those assumptions are indeed satisfied under the linear regression model as shown in Lemma \ref{lmm:least-squares}.
%
\end{remark}

\subsection{Main Theorem}\label{sec:main_formal}

By combining  Theorem \ref{thm:convg_BGD}, Lemma \ref{lmm:geometric_median}, and  Theorem \ref{thm:unif_convg},
we prove the main theorem.

\begin{theorem}\label{thm:main}
Suppose Assumption \ref{ass:pop_risk_smooth}  -- Assumption \ref{ass:bounded_hessian} hold,
and $\Theta \subset  \{ \theta: \| \theta-\theta^*\| \le r \sqrt{d} \}$ for some positive parameter $r$.
Assume $2(1+\epsilon)q \le k \le m$. Fix any constant $\alpha \in (q/k,1/2)$ and any $\delta>0$ such that $\delta \le \alpha -q/k$.
If $\Delta_1 (N/k) \le \sigma_1^2/\alpha_1,$  $\Delta_2 (N/k) \le \sigma_2^2/\alpha_2 $, and
\begin{align*}
\rho = 1- \sqrt{1-L^2/(4M^2)} -   \xi_2 L/(2M^2) >0
\end{align*}
for $\xi_2=8 C_\alpha \times \Delta_2 (N/k)$,
then with probability at least
$$1- \exp (- k D(\alpha-q/k \| \delta)  ),$$
 the iterates $\{\theta_t\}$
given by Algorithm \ref{alg:BGD} with $\eta=L/(2M^2)$ satisfy
$$
\| \theta_t - \theta^* \| \le \left( 1- \rho \right)^t \| \theta_0 - \theta^* \| +   \eta\xi_1/\rho, \quad \forall t \ge 1,
$$
where $\xi_1=4 C_\alpha \times \Delta_1 (N/k)$.
\end{theorem}
Under certain conditions, we are able to further bound $\xi_1$ and $\xi_2$. Next we present a formal statement of Theorem \ref{thm:main_informal}; it readily follows from Theorem \ref{thm:main}.

\begin{theorem}\label{cor:main}
Suppose that Assumption \ref{ass:pop_risk_smooth} -- Assumption \ref{ass:bounded_hessian} hold
such that $L$, $M$,  $\sigma_1,$ $\alpha_1,$ $\sigma_2,$ $\alpha_2$ are all of $\Theta(1)$,
and $\log M' = O\left(\log d \right).$
Assume that
 $\Theta \subset  \{ \theta: \| \theta-\theta^*\| \le r \sqrt{d} \}$ for some positive parameter $r$
such that $\log (r)= O(d \log (N/k))$, and $2(1+\epsilon)q \le k \le m$.
Fix any $\alpha \in (q/k,1/2)$ and any $\delta>0$ such that $\delta \le \alpha -q/k$ and
$\log(1/\delta)= O(d)$.
There exist universal positive constants $c_1, c_2$ such that
if
$$
\frac{N}{k}  \ge c_1 C_\alpha^2 d \log (N/k),$$
then
with probability at least
$$1- \exp (- k D(\alpha-q/k \| \delta)  ),
$$
 the iterates $\{\theta_t\}$
given by Algorithm \ref{alg:BGD} with $\eta=L/(2M^2)$ satisfy
$$
\| \theta_t - \theta^* \| \le \left( \frac{1}{2} + \frac{1}{2} \sqrt{ 1- \frac{L^2}{4M^2} }\right)^t \| \theta_0 - \theta^* \| + c_2 \sqrt{ \frac{dk}{N} }, \quad \forall t \ge 1.
$$
\end{theorem}
\begin{proof}
Recall from \eqref{eq:s} that
\begin{align*}
\Delta_1 (N/k, d, \delta, \sigma_1)=\sqrt{2} \sigma_1 \sqrt{ \frac{d \log 6 + \log(3/\delta) }{ N/k}  }.
\end{align*}
When $\sigma_1=\Theta(1)$ and $\log (1/\delta) =O(d)$, it holds that $\Delta_1 (N/k) =\Theta\left( \sqrt{ kd /N} \right)$.
Similarly,  we have  $\Delta_2(N/k) =   \Theta \left( \sqrt{\frac{kd \log (N/k)}{ N} } \right) $.
Hence, there exists an universal
positive constant
$c_1$ such that for all $N/k \ge c_1 C_\alpha^2 d \log (N/k)$, it holds that
$\Delta_1(N/k) \le \sigma_1^2/\alpha_1$, $\Delta_2(N/k) \le \sigma_2^2/\alpha_2$, and
\begin{align}
 \Delta_2 (N/k) \le \frac{M^2}{8C_\alpha L} \left( 1- \sqrt{1 - L^2 /(4M^2) } \right). \label{eq:rho_condition}
\end{align}
Thus we have
$$
\xi_2=8 C_\alpha \times \Delta_2 (N/k) \le \frac{M^2}{L} \left( 1- \sqrt{1 - L^2 /(4M^2) } \right),
$$
and as a consequence,
$$
\rho = 1- \sqrt{1-L^2/(4M^2)} - \frac{\xi_2 L}{ 2M^2 } \ge \frac{1}{2} - \frac{1}{2}  \sqrt{1-L^2/(4M^2)} >0.
$$
Hence, we can apply Theorem \ref{thm:main}. Finally, to finish the proof,
recall that $\eta= L/(2M^2)$ and $\xi_1=4 C_\alpha \times \Delta_1(N/k)$; thus
the term $\eta \xi_1/\rho$ can be bounded as follows:
\begin{align*}
\frac{\eta \xi_1}{\rho}   = \frac{L}{2M^2} \times  \frac{4C_{\alpha} \Delta_1 (N/k)}{\rho }
\le  \frac{L}{2M^2} \times \frac{8C_{\alpha} \Delta_1 (N/k)}{1 - \sqrt{1-L^2/(4M^2)}}  \le c_2 \sqrt{\frac{dk}{N}},
\end{align*}
where $c_2$ is some universal constant.
\end{proof}

\section{Application to Linear Regression}
\label{sec: linear regression}
We illustrate our general results by applying them to
the classical linear regression problem. Let $ X_i = (w_i, y_i) \in \reals^{d} \times \reals $
denote
the input data and define the risk function
$
f(X_i, \theta) = \frac{1}{2} \left( \langle w_i, \theta \rangle - y_i \right)^2.
$
For simplicity, we assume that $y_i$ is indeed generated from a linear model:
$$
y_i = \langle w_i, \theta^* \rangle + \zeta_i ,
$$
where $ \theta^* $ is an unknown true model parameter,
$ w_i \sim N(0, \identity) $ is the covariate vector whose
covariance matrix is assumed to be identity,  and $ \zeta_i \sim N(0,1) $ is i.i.d.\ additive Gaussian noise
independent of $w_i$'s. Intuitively, the inner product $ \langle w_i, \theta^* \rangle $ can be viewed as some ``measurement" of $\theta^*$ -- the signal; and $\zeta_i $ is the additive noise.

The population risk minimization problem \eqref{eq:min_pop_risk} is simply
\begin{align*}
\min_\theta  ~~ \frac{1}{2}\| \theta - \theta^* \|_2^2 + \frac{1}{2},
\end{align*}
where
\begin{align*}
F(\theta) &\triangleq \expect{f(X,\theta)} = \expect{  \frac{1}{2} \left( \langle w, \theta \rangle - y \right)^2}\\
&=  \expect{  \frac{1}{2} \left( \langle w, \theta \rangle - \langle w, \theta^* \rangle - \zeta \right)^2}
= \frac{1}{2}\| \theta - \theta^* \|_2^2 + \frac{1}{2},
\end{align*}
for which $ \theta^* $ is indeed the unique minimum. If the function $F(\cdot)$ can be computed exactly, then $\theta^*$ can be read from its expression directly. The standard gradient descent method for minimizing $F(\cdot)$ is also straightforward. The population gradient is
$\nabla_\theta F(\theta) = \theta - \theta^*$. It is easy to see that the population risk $F$ is $M$-Lipschitz continuous with $M=1$, and $L$-strongly convex with $L=1$.  Hence, Assumption \ref{ass:pop_risk_smooth} is satisfied with $M=L=1$; and the stepsize $\eta=L/(2M^2) = 1/2$.

In practice, unfortunately, since we do not know exactly the distribution of the random input $X$, we can neither read $\theta^*$ from the expression $ F(\cdot)$ nor compute the population gradient $\nabla F(\theta)$ exactly. We are only able to approximate the population risk $F(\cdot)$ or the population gradient $\nabla F(\theta)$. Our focus is the gradient approximation. In particular, for a given random sample, the associated random gradient is given by
$
\nabla f(X, \theta) =w \langle w, \theta-\theta^* \rangle - w \zeta,
$
where $w \sim \calN(0, \identity)$ and $\zeta \sim \calN(0,1)$ that is independent of $w$. 

%

The following lemma verifies that Assumption \ref{ass:gradient_sub_gaussian}--Assumption \ref{ass:bounded_hessian}
are satisfied with appropriate parameters. 
\begin{lemma}\label{lmm:least-squares}
Under the linear regression model, the sample gradient function $\nabla f(X, \cdot)$ satisfies\\
(1) Assumption \ref{ass:gradient_sub_gaussian} with $\sigma_1=\sqrt{2}$
and $\alpha_1=\sqrt{2}$, \\
(2) Assumption \ref{ass:gradient_sub_exp} with $\sigma_2=\sqrt{8}$ and $\alpha_2=8$, \\
(3) and Assumption \ref{ass:bounded_hessian} with $M'(\delta)=d + 2\sqrt{d \log(4/\delta)} + 2 \log(4/\delta).$
\end{lemma}
\begin{proof}

We first check Assumption \ref{ass:gradient_sub_gaussian}.
Recall that $\nabla f(X, \theta) =w \langle w, \theta-\theta^* \rangle - w \zeta$,
where $w \sim \calN(0, \identity)$ and $\zeta \sim \calN(0,1)$ is independent of $w$.
Hence, $\nabla f(X, \theta^*) = - w \zeta$.
It follows that for any $v$ in unit sphere $B$,
$$\iprod{\nabla f(X, \theta^*) } { v } = -\zeta \iprod{w}{v}.$$
Because $w \sim \calN(0, \identity)$ and are independent of $\zeta$, it holds that $\iprod{w}{v} \sim \calN(0, 1)$ and is independent of $\zeta$. Thus, to compute $\expect{ \exp \left( -\lambda \zeta \iprod{w}{v} \right) }$, we can use the standard conditioning argument. In particular, for $\lambda^2<1$,
\begin{align}
\label{conditioning}
\nonumber
\expect{\exp \left( \lambda \iprod{\nabla f(X, \theta^*) } { v } \right) }
\nonumber
&=\expect{ \exp \left( -\lambda \zeta \iprod{w}{v} \right) }\\
&=\expect{ \expect{\exp \left( -\lambda y \iprod{w}{v} \right) |\zeta=y } },
\end{align}
where the expectation of $\expect{\exp \left( -\lambda y \iprod{w}{v} \right) |\zeta=y } $ is taken over the conditional distribution of $ \iprod{w}{v}$ conditioning on $\zeta$  being $y$. Since $\iprod{w}{v}$ and $\zeta$ are independent of each other, the conditional distribution of $ \iprod{w}{v}$ w.\ r.\ t.\ $\zeta$ is the same as the unconditional distribution of $ \iprod{w}{v}$, which is a Gaussian distribution. Thus, we can apply the moment generating function of Gaussian distribution to get
\begin{align*}
\expect{\exp \left( -\lambda y \iprod{w}{v} \right) |\zeta=y } = \exp \left( \lambda^2  y^2 /2 \right).
\end{align*}
Then, the right-hand side of \eqref{conditioning} becomes
\begin{align}
\label{conditioning 1}
\nonumber
\expect{\exp \left( \lambda \iprod{\nabla f(X, \theta^*) } { v } \right) }
\nonumber
&=\expect{ \expect{\exp \left( -\lambda y \iprod{w}{v} \right) |\zeta=y }}\\
\nonumber
&=\expect{\exp \left( \lambda^2 {\zeta}^2 /2 \right)  }\\
&\overset{(a)}{=}(1-\lambda^2)^{-1/2},
\end{align}
%
where equality $(a)$ follows from the moment generating function of $\chi^2$ distribution, \ie,
$$\expect{\exp \left( t \zeta^2 \right)  } = (1-2t)^{-1/2}~~~ \text{for}~~t < 1/2. $$

Using the fact that $1-\lambda^2 \ge e^{-2\lambda^2}$ for $\lambda^2 \le 1/2$,
it follows that
$$
\expect{\exp \left( \lambda \iprod{\nabla f(X, \theta) } { v } \right) } \le e^{\lambda^2}, \quad \forall |\lambda| \le \frac{1}{\sqrt{2}}.
$$
Thus Assumption \ref{ass:gradient_sub_gaussian} holds with $\sigma_1=\sqrt{2}$
and $\alpha_1=\sqrt{2}$.

Next, we verify Assumption \ref{ass:bounded_hessian}.
Note that $\nabla^2 f(X, \theta) =ww^\top$ and hence it suffices to
show that
\begin{align*}
 \prob{ \left\|\frac{1}{n} \sum_{i=1}^n \nabla^2 f(X_i, \theta)  \right\| \le M' } &=\prob{ \left\| \frac{1}{n} \sum_{i=1}^n  w_i w_i^\top \right\| \le M' } \ge 1- \frac{\delta}{3},
\end{align*}
for some $M'$ depending on $n,$ $d,$ and $\delta$.

Let $W=[w_1, w_2, \ldots, w_n]$ denote the $d \times n$ matrix
whose columns are given by $w_i$'s. Then
$\sum_{i=1}^n  w_i w_i^\top = W W^\top$.
Also, the spectral norm of $WW^\top$ equals $\|W\|^2$.
Therefore,
$$
\prob{ \left\| \frac{1}{n} \sum_{i=1}^n  w_i w_i^\top \right\| \le M' }
=\prob{ \left\| W \right\| \le \sqrt{n M'} }.
$$
Note that $W$ is an $d \times n$ matrix with i.i.d.\ standard Gaussian entries.
Standard Gaussian matrix concentration inequality (see, e.g., \cite[Corollary 5.35]{vershynin2010nonasym}) states
that for every $t \ge 0$,
$$
\prob{ \| W\| \le \sqrt{n} + \sqrt{d} + t } \ge 1- \exp(-t^2/2).
$$
Plugging $t = \sqrt{2 \log (4/\delta)}$  and setting
$$
M' = \frac{1}{n} \left( \sqrt{n} + \sqrt{d} + \sqrt{ 2 \log (4/\delta)} \right)^2
$$
complete the proof.


Finally, we verify Assumption \ref{ass:gradient_sub_exp}. Recall
that the gradient difference $h(X,\theta)$ is given by
$
h(X, \theta)= w \langle w, \theta-\theta^* \rangle,
$
and $\expect{h(X, \theta)}=\theta-\theta^*$.
It follows that for any vector $v$ in unit sphere $B$,
$$
\iprod{h(X, \theta)-\expect{ h(X, \theta)}}{ v} = \Iprod{w} {\theta-\theta^*}  \Iprod{w}{v} - \Iprod{\theta-\theta^*}{v}.
$$
For a fixed $\theta \in \Theta$ with $\theta \neq \theta^*$ and let $\tau=\| \theta-\theta^\ast\|>0$.
Then we have the following orthogonal decomposition:
$\theta-\theta^* = \sqrt{\gamma} v + \sqrt{\eta}v_{\perp}$,
where $\gamma + \eta = \tau^2$, and $v_{\perp}$ denote an vector perpendicular to $v$.
It follows that
$$
\Iprod{w} {\theta-\theta^*}  \Iprod{w}{v}  - \Iprod{\theta-\theta^*}{v}
=\sqrt{\gamma} \Iprod{w}{v}^2 -\sqrt{\gamma} + \sqrt{\eta} \Iprod{w}{v_{\perp}} \Iprod{w}{v}.
$$
It is easy to see that random variables $ \Iprod{w}{v_{\perp}} \sim \calN(0,1)$ and $\Iprod{w}{v} \sim \calN(0,1)$ are jointly Gaussian. In addition, we have
\begin{align*}
\expect{\Iprod{w}{v_{\perp}}\Iprod{w}{v}} &= \expect{v_{\perp}^\top w w^\top v} \\
&= v_{\perp}^\top \expect{w w^\top}  v = v_{\perp}^\top \mathbf{I}  v =0.
\end{align*}
Thus, $ \Iprod{w}{v_{\perp}} \sim \calN(0,1)$ and $\Iprod{w}{v} \sim \calN(0,1)$ are mutually independent.

For any $\lambda$ with $\lambda \sqrt{\gamma}<1/4$ and $\lambda^2\eta<1/4$,
\begin{align*}
&\expect{ \exp \left(  \lambda \Iprod{h(X, \theta)-\expect{ h(X, \theta)} }{v} \right) } \\
&=
\expect{\exp \left( \lambda \sqrt{\gamma} \left( \Iprod{w}{v}^2 -1 \right) +  \lambda \sqrt{\eta} \Iprod{w}{v_{\perp}} \Iprod{w}{v}  \right) }  \\
& \le \sqrt{ \expect{ e^{2 \lambda   \sqrt{\gamma} \left(  \Iprod{w}{v}^2  -1 \right) } } \expect{ e^{2 \lambda  \sqrt{\eta} \Iprod{w}{v_{\perp}} \Iprod{w}{v} } } } \\
&=e^{-\lambda \sqrt{\gamma}} \sqrt{ \expect{ e^{2 \lambda   \sqrt{\gamma} \left(  \Iprod{w}{v}^2 \right) }}} \sqrt{\expect{ e^{2 \lambda  \sqrt{\eta} \Iprod{w}{v_{\perp}} \Iprod{w}{v} } } }\\
& =  e^{-\lambda \sqrt{\gamma}} \left( 1 - 4 \lambda \sqrt{\gamma} \right)^{-1/4}   \left( 1 - 4 \lambda^2 \eta  \right)^{-1/4},
\end{align*}
where the first inequality holds due to Cauchy-Schwartz's inequality, and the last equality follows by
plugging in the moment generating functions for $\chi^2$ distributions as well as using the conditioning argument that is similar to the derivation of \eqref{conditioning}.

Using the fact that $e^{-t}/\sqrt{1-2t} \le e^{2t^2}$ for $|t| \le 1/4$
and $1-t \ge e^{-4t}$ for $0 \le t \le 1/2$ ,
it follows that for $\lambda^2 \le 1/(64\tau^2)$,
\begin{align*}
\expect{ \exp \left(  \lambda \Iprod{h(X, \theta)-\expect{ h(X, \theta)} }{v} \right) }
& \le \exp \left( 4\lambda^2 (\gamma+ \eta) \right) \\
& \le \exp \left( 4\lambda^2 \tau^2 \right).
\end{align*}
Hence, Assumption \ref{ass:gradient_sub_exp} holds with $\sigma_2=\sqrt{8}$ and $\alpha_2=8$.

\end{proof}

Thus, according to Theorem \ref{thm:main_informal}, our Byzantine Gradient Descent method can robustly solve the linear regression problem exponentially fast with high probability -- formally stated the following corollary.

\begin{corollary}[Linear regression]
\label{linear regression conv}
Under the aforementioned least-squares model for linear regression, assume  $\Theta \subset  \{ \theta: \| \theta-\theta^*\| \le r\sqrt{d} \}$ for $r>0$ such that  $\log r = O(d \log (N/k))$.
Suppose that $2(1+\epsilon)q \le k \le m $.
Fix any  $\alpha \in (q/k,1/2)$ and any $\delta>0$ such that $\delta \le \alpha-q/k$ and $\log(1/\delta)=O(d)$, there exist universal constants $c_1, c_2>0$ such that if $N/k \ge c_1 C_\alpha^2 d \log (N/k)$. Then with probability at least $1-\exp(- k D( \pth{\alpha -q/k} \| \delta ) )$, the iterates $\{\theta_t\}$
given by Algorithm \ref{alg:BGD} with $\eta=1/2$ satisfy
$$
\| \theta_t - \theta^* \| \le \left(\frac{1}{2} + \frac{\sqrt{3}}{4} \right)^t \| \theta_0 - \theta^* \| + c_2 C_\alpha \sqrt{ \frac{dk}{N} }, \quad \forall t \ge 1.
$$
\end{corollary}

Note that in Corollary \ref{linear regression conv}, we assume the ``searching space" $\Theta$ belongs to some range, which may grow with $d$ and $N/k$. This assumption is rather mild since in practice; we typically do have some prior knowledge about the range of $\theta^*$.


\section{Related Work}
\label{sec: related work}
The present paper intersects with two main areas of research:
statistical machine learning and distributed computing.
Most related to our work is \cite{Blanchard2017} that we became aware of when preparing this paper. It also studies distributed optimization in adversarial settings, but the setup is different from ours. In particular, their focus is solving an optimization problem, where all
$m$ working machines have access to a common dataset $\{x_i\}_{i=1}^N$
and the goal is to collectively compute the minimizer $\hat{\theta}$ of the average
cost $Q(\theta) = (1/N) \sum_{i=1}^N f(x_i, \theta)$.  Importantly, the dataset $\{x_i\}_{i=1}^N$ are assumed to be deterministic.
In contrast, we adopt the standard statistical learning framework, where each working machine only has access to its own data samples, which are assumed to be generated from some unknown distribution $\mu$, and the goal is to estimate the optimal model parameter $\theta^*$ that minimizes the true prediction error
$\mathbb{E}_{X \sim \mu} [ f(X, \theta)]$ --- as mentioned, characterizing the statistical estimation accuracy is a main focus of ours. Our algorithmic approaches and main results are also significantly different.
The almost sure convergence is proved in~\cite{Blanchard2017} without an explicit characterization of convergence speed nor the estimation errors.

Our work is also closely related to the literature on
robust parameter estimation using geometric median.
It is shown in~\cite{lopuhaa1991breakdown} that geometric median has a breakdown point of $0.5$, that is, given a collection
of $n$ vectors in $\reals^d$, at least $\lfloor (n+1)/2 \rfloor /n$ number of points
needs to be corrupted in order to arbitrarily perturb the geometric median. A more quantitative
robustness result is recently derived in~\cite[Lemma 2.1]{minsker2015geometric}.
The geometric median has been applied to distributed machine learning under the one-shot aggregation framework~\cite{feng2014distributed},  under the restrictive
assumption that the number of data available in each working machine satisfies
$N/m \gg d$.  While we also apply geometric median-of-mean as a sub-routine, our problem setup, overall algorithms and main results are completely different.

A recent line of work~\cite{diakonikolas2016robust,lai2016agnostic} presents polynomial algorithms to consistently estimate the mean and
 covariance of a distribution
from $N$ i.i.d.\ samples in $\reals^d$, in the presence of an $\epsilon$ fraction of malicious errors for sufficiently small $\epsilon$,
while geometric median is proved to fail when $\epsilon = \Omega(1/\sqrt{d})$.  However, it is unclear how to directly apply their results
to our gradient descent setting, where our goal is to robustly estimate a $d$-dimensional gradient function from i.i.d. sample gradient
functions.

On the technical front, a crucial step in our convergence proof is to show the
geometric median of means of $n$ i.i.d.\  random gradients converges to the underlying gradient function $\nabla F(\theta)$ uniformly over $ \theta $. Our proof builds on several ideas from the empirical process theory, which guarantees uniform convergence of
the empirical risk function $(1/n)\sum_{i=1}^n f(X_i, \cdot)$ to the population risk $F(\cdot)$.
However, what we need is the uniform convergence of empirical \emph{gradient} function
$(1/n) \sum_{i=1}^n \nabla f(X_i, \cdot)$, as well as its \emph{geometric median} version, to the population gradient function $\nabla F(\cdot)$.
To this end, we use concentration inequalities to first establish point-wise convergence
and then boost it to uniform convergence via the celebrated $\epsilon$-net argument. Similar ideas have been used recently in the work~\cite{mei2016landscape}, which studies the stationary points of the empirical risk function.

\section{Discussion}
\label{sec: discussion}

In this paper, we consider the machine learning scenario where the model is trained in an unsecured environment. As a result of this, the model training procedure needs to be robust to adversarial interruptions. Based on the geometric median of means, we propose a communication-efficient and robust method for the parameter server to aggregate the gradients reported by the unreliable workers. In each iteration, the parameter server first groups the received gradients into non-overlapping batches to increase the ``similarity" of the Byzantine-free batches; and then takes the median of the batch gradients to cripple the interruption of Byzantine machines.

There are many other interesting directions. We list a few of them as follows.
\begin{itemize}
\item As mentioned in the introduction, Federated Learning is proposed due to the users' concerns about privacy breaches. In Federated Learning, the training data is kept locally on user's devices, which indeed grants users the control of their data. Nevertheless, to have a high-quality model trained, information about their data need to be extracted. In our future work, we would like to provide a precise characterization of the minimal amount of privacy has to be sacrificed in the Federated Learning paradigm.

\item In addition to security, low volume of local data and communication constraints, there are many other practical challenges such as intermittent availability of mobile phones, i.e., communication asynchrony. Although our algorithm only needs $\log (N)$ rounds, a single synchronous round may be significantly ``delayed" by the slow machines. We would like to adapt our algorithms to the asynchronous setting.

\item In Byzantine fault models, we assume the Byzantine adversaries know the realization of the random bits generated by the parameter server.
Depending on the applications, this assumption can possibly be relaxed, which may lead to simpler algorithms.
A simple idea to defend against the relaxed Byzantine faults is to selects a subset of received gradients at each iteration and then takes the average over the selected gradients.
One selection rule is random selection and another one is to select the gradients of the small $\ell_2$ norms. It would be interesting to investigate the performance of
these two selection rules and compare them with the geometric median.

\end{itemize}

\appendix

\section{Proof of Lemma 2.1}
\label{pf_geometric_median_robust}
\begin{proof}
Let $S=\{i : \| z_i \| \le r\}$. For any $i \in S$, we have
$$
\| z_\ast  - z_i \| \ge \|z_\ast \| - \|z_i \| \ge \|z_\ast \| - 2 r + \|z_i \|.
$$
Moreover, by triangle's inequality, for all $i \notin S$, we have
$$
\| z_\ast - z_i \| \ge \| z_i \| - \| z_\ast \|
$$
Combining the last two displayed equations yields that
$$
\sum_{i=1}^n \| z_\ast - z_i \| \ge \sum_{i=1}^n \|z_i \| + (2 |S| - n ) \|z_\ast\| -2 |S| r.
$$
Since $z_\ast$ is a $(1+\gamma)$-approximate solution of $
\sum_{i=1}^n \| z -z_i \|$, it follows that
$$
\sum_{i=1}^n \|z_i \| + (2 |S| - n ) \|z_\ast\| -2 |S| r \le (1+\gamma) \min_{z} \| z - z_i \|.
$$
Note that $\sum_{i=1}^n \|z_i \|  = \sum_{i=1}^n \|0-z_i \| \ge  \min_{z} \sum_{i=1}^n  \| z - z_i \|$. Hence, it further
implies that
$$
(2 |S| - n ) \|z_\ast\| -2 |S| r \le \gamma  \min_{z}\sum_{i=1}^n   \| z - z_i \|,
$$
and thus
\begin{align*}
 \|z_\ast\| & \le \frac{2 |S| r}{2|S| - n} + \frac{\gamma   \min_{z} \sum_{i=1}^n  \| z - z_i \| }{ 2 |S| - n} \le \frac{2(1-\alpha) r }{ 1-2\alpha } + \frac{\gamma  \min_{z} \sum_{i=1}^n  \| z - z_i \| }{(1-2\alpha) n},
\end{align*}
where the last inequality holds due to $|S| \ge (1-\alpha)n$ by the assumption.
\end{proof}

\section{Proof of Lemma 3.2}
\label{pf app convergence}
\begin{proof}
By \eqref{eq:pop_gd} and the fact that $\nabla F(\theta^*)=\vct{0}$, we have
\begin{align*}
\| \theta' - \theta^*\|^2&=\| \theta - \theta^* - \eta  \nabla F(\theta) \|^2\\
& =\left\| \theta - \theta^* - \eta \left( \nabla F(\theta) - \nabla F(\theta^*) \right) \right\|^2 \\
& = \left\| \theta - \theta^* \right\|^2 + \eta^2 \left\| \nabla F(\theta) - \nabla F(\theta^*) \right\|^2 - 2\eta \iprod{\theta - \theta^*}{ \nabla F(\theta) - \nabla F(\theta^*)}.
\end{align*}
By Assumption \ref{ass:pop_risk_smooth}, we have
$$
\left\| \nabla F(\theta) - \nabla F(\theta^*) \right\| \le M \| \theta - \theta^* \|,
$$
$$
F(\theta) \ge  F(\theta^*) + \iprod{\nabla F(\theta^*)}{\theta - \theta^*} + \frac{L}{2} \| \theta - \theta^* \|^2,
$$
and
$$
F(\theta^*) \ge F(\theta) + \iprod{\nabla F(\theta)}{ \theta^*-\theta}.
$$
Summing up the last two displayed equations yields that
$$
0 \ge \iprod{\nabla F(\theta) - \nabla F(\theta^*) }{ \theta^*-\theta} + \frac{L}{2} \| \theta - \theta^* \|^2.
$$
Therefore,
$$
\| \theta' - \theta^*\|^2
\le \left( 1+ \eta^2 M^2 - \eta L \right) \left\| \theta - \theta^* \right\|^2.
$$
The conclusion follows by the choosing $\eta=L/2M^2$.
\end{proof}

\section{Concentration Inequality for Sub-exponential Random Variables}
\label{subexpon prelim}
\begin{definition}[Sub-exponential]
	Random variable X with mean $\mu$ is sub-exponential if $\exists\ \nu > 0$ and $\alpha > 0 $ such that
	$$
	\expect{ \exp \left( \lambda(X-\mu) \right) } \leq \exp \left( \frac{ \nu^2 \lambda^2}{2} \right), \quad  \forall |\lambda| \leq \frac{1}{\alpha}.
	$$
\end{definition}
\begin{theorem} \label{thm:sum_sub_exponential}
If $X_{1}, \ldots, X_{n}$ are independent random variables where $X_{i}$'s are sub-exponential
with scaling parameters $(\nu_{i}, \alpha_{i})$ and mean $\mu_{i}$, then $\sum_{i=1}^{n}X_{i}$ is sub-exponential with scaling parameters $(\nu_{\ast},\alpha_{\ast})$, where $\nu_{\ast}^{2} = \sum_{i=1}^{n}\nu_{i}^{2}$ and $\alpha_{\ast} = max_{1\leq i \leq n}\alpha_{i}$. Moreover,
\begin{align*}
\prob{ \sum_{i=1}^{n} \left(X_{i}- \mu_{i} \right) \geq t } \leq
	\begin{cases}
		\exp \left( -t^{2} /( 2\nu_{\ast}^{2} ) \right)  & \text{ if }  0 \le t \le \nu_{\ast}^{2}/\alpha_{\ast} \\
		\exp \left( -t/ ( 2\alpha_{\ast} ) \right) &  \text{ o.w. }	
		\end{cases}
\end{align*}
\end{theorem}

\section*{Acknowledgement}
Y. Chen was partially supported by the National Science Foundation under CRII award 1657420 and grant CCF-1704828, and by the School of Operations Research and Information Engineering at Cornell University. L. Su was partially supported by the National Science Foundation Grant ECCS-1610543 and NSF Science \& Technology Center for Science of Information Grant CCF-0939370.

\bibliographystyle{alpha}
\bibliography{RobGrad}

\newcommand{\etalchar}[1]{$^{#1}$}
\begin{thebibliography}{WWRL10}

\bibitem[AS00]{agrawal2000privacy}
Rakesh Agrawal and Ramakrishnan Srikant.
\newblock Privacy-preserving data mining.
\newblock {\em SIGMOD Rec.}, 29(2):439--450, May 2000.

\bibitem[BMGS17]{Blanchard2017}
Peva Blanchard, El~Mahdi~El Mhamdi, Rachid Guerraoui, and Julien Stainer.
\newblock Byzantine-tolerant machine learning.
\newblock {\em arXiv preprint arXiv:1703.02757}, 2017.

\bibitem[BPC{\etalchar{+}}11]{boyd2011distributed}
Stephen Boyd, Neal Parikh, Eric Chu, Borja Peleato, and Jonathan Eckstein.
\newblock Distributed optimization and statistical learning via the alternating
  direction method of multipliers.
\newblock {\em Foundations and Trends{\textregistered} in Machine Learning},
  3(1):1--122, 2011.

\bibitem[BV04]{boyd2004convex}
Stephen Boyd and Lieven Vandenberghe.
\newblock {\em Convex optimization}.
\newblock Cambridge university press, 2004.

\bibitem[CCZ{\etalchar{+}}13]{cardot2013efficient}
Herv{\'e} Cardot, Peggy C{\'e}nac, Pierre-Andr{\'e} Zitt, et~al.
\newblock Efficient and fast estimation of the geometric median in hilbert
  spaces with an averaged stochastic gradient algorithm.
\newblock {\em Bernoulli}, 19(1):18--43, 2013.

\bibitem[CLM{\etalchar{+}}16]{cohen2016geometric}
Michael~B Cohen, Yin~Tat Lee, Gary Miller, Jakub Pachocki, and Aaron Sidford.
\newblock Geometric median in nearly linear time.
\newblock In {\em Proceedings of the 48th Annual ACM SIGACT Symposium on Theory
  of Computing}, pages 9--21. ACM, 2016.

\bibitem[DG08]{dean2008mapreduce}
Jeffrey Dean and Sanjay Ghemawat.
\newblock Mapreduce: simplified data processing on large clusters.
\newblock {\em Communications of the ACM}, 51(1):107--113, 2008.

\bibitem[DKK{\etalchar{+}}16]{diakonikolas2016robust}
Ilias Diakonikolas, Gautam Kamath, Daniel~M Kane, Jerry Li, Ankur Moitra, and
  Alistair Stewart.
\newblock Robust estimators in high dimensions without the computational
  intractability.
\newblock In {\em Foundations of Computer Science (FOCS), 2016 IEEE 57th Annual
  Symposium on}, pages 655--664. IEEE, 2016.

\bibitem[DWJ13]{duchi2013local}
John Duchi, Martin~J Wainwright, and Michael~I Jordan.
\newblock Local privacy and minimax bounds: Sharp rates for probability
  estimation.
\newblock In {\em Advances in Neural Information Processing Systems}, pages
  1529--1537, 2013.

\bibitem[FXM14]{feng2014distributed}
Jiashi Feng, Huan Xu, and Shie Mannor.
\newblock Distributed robust learning.
\newblock {\em arXiv preprint arXiv:1409.5937}, 2014.

\bibitem[JLY16]{jordan2016communication}
Michael~I Jordan, Jason~D Lee, and Yun Yang.
\newblock Communication-efficient distributed statistical inference.
\newblock {\em arXiv preprint arXiv:1605.07689}, 2016.

\bibitem[Kem87]{kemperman1987median}
JHB Kemperman.
\newblock The median of a finite measure on a banach space.
\newblock {\em Statistical data analysis based on the L1-norm and related
  methods (Neuch{\^a}tel, 1987)}, pages 217--230, 1987.

\bibitem[KMR15]{konevcny2015federated}
Jakub Kone{\v{c}}n{\`y}, Brendan McMahan, and Daniel Ramage.
\newblock Federated optimization: Distributed optimization beyond the
  datacenter.
\newblock {\em arXiv preprint arXiv:1511.03575}, 2015.

\bibitem[KPS02]{kaufman2002network}
Charlie Kaufman, Radia Perlman, and Mike Speciner.
\newblock {\em Network security: private communication in a public world}.
\newblock Prentice Hall Press, 2002.

\bibitem[LBG{\etalchar{+}}12]{low2012distributed}
Yucheng Low, Danny Bickson, Joseph Gonzalez, Carlos Guestrin, Aapo Kyrola, and
  Joseph~M Hellerstein.
\newblock Distributed graphlab: a framework for machine learning and data
  mining in the cloud.
\newblock {\em Proceedings of the VLDB Endowment}, 5(8):716--727, 2012.

\bibitem[LR91]{lopuhaa1991breakdown}
Hendrik~P Lopuhaa and Peter~J Rousseeuw.
\newblock Breakdown points of affine equivariant estimators of multivariate
  location and covariance matrices.
\newblock {\em The Annals of Statistics}, pages 229--248, 1991.

\bibitem[LRV16]{lai2016agnostic}
Kevin~A Lai, Anup~B Rao, and Santosh Vempala.
\newblock Agnostic estimation of mean and covariance.
\newblock In {\em Foundations of Computer Science (FOCS), 2016 IEEE 57th Annual
  Symposium on}, pages 665--674. IEEE, 2016.

\bibitem[Lyn96]{Lynch:1996:DA:2821576}
Nancy~A. Lynch.
\newblock {\em Distributed Algorithms}.
\newblock Morgan Kaufmann Publishers Inc., San Francisco, CA, USA, 1996.

\bibitem[M{\etalchar{+}}15]{minsker2015geometric}
Stanislav Minsker et~al.
\newblock Geometric median and robust estimation in banach spaces.
\newblock {\em Bernoulli}, 21(4):2308--2335, 2015.

\bibitem[MBM16]{mei2016landscape}
Song Mei, Yu~Bai, and Andrea Montanari.
\newblock The landscape of empirical risk for non-convex losses.
\newblock {\em arXiv preprint arXiv:1607.06534}, 2016.

\bibitem[MD{\etalchar{+}}87]{milasevic1987uniqueness}
P~Milasevic, GR~Ducharme, et~al.
\newblock Uniqueness of the spatial median.
\newblock {\em The Annals of Statistics}, 15(3):1332--1333, 1987.

\bibitem[MNO{\etalchar{+}}10]{mottonen2010asymptotic}
Jyrki M{\"o}tt{\"o}nen, Klaus Nordhausen, Hannu Oja, et~al.
\newblock Asymptotic theory of the spatial median.
\newblock In {\em Nonparametrics and Robustness in Modern Statistical Inference
  and Time Series Analysis: A Festschrift in honor of Professor Jana
  Jure{\v{c}}kov{\'a}}, pages 182--193. Institute of Mathematical Statistics,
  2010.

\bibitem[MNSJ15]{moritz2015sparknet}
Philipp Moritz, Robert Nishihara, Ion Stoica, and Michael~I Jordan.
\newblock Sparknet: Training deep networks in spark.
\newblock {\em arXiv preprint arXiv:1511.06051}, 2015.

\bibitem[MR10]{federatedlearningblog}
Brendan McMahan and Daniel Ramage.
\newblock Federated learning: Collaborative machine learning without
  centralized training data.
\newblock {\em
  \url{https://research.googleblog.com/2017/04/federated-learning-collaborative.html}},
  Accessed: 2017-04-10.

\bibitem[PH96]{provost1996scaling}
Foster~J Provost and Daniel~N Hennessy.
\newblock Scaling up: Distributed machine learning with cooperation.
\newblock In {\em AAAI/IAAI, Vol. 1}, pages 74--79. Citeseer, 1996.

\bibitem[PP02]{Pfleeger:2002:SC:579149}
Charles~P. Pfleeger and Shari~Lawrence Pfleeger.
\newblock {\em Security in Computing}.
\newblock Prentice Hall Professional Technical Reference, 3rd edition, 2002.

\bibitem[SB98]{sherali1998network}
Hanif~D Sherali and Dimitri~P Bertsekas.
\newblock Network optimization: Continuous and discrete models, 1998.

\bibitem[Ver10]{vershynin2010nonasym}
R.~Vershynin.
\newblock Introduction to the non-asymptotic analysis of random matrices.
\newblock {\em Arxiv preprint arxiv:1011.3027}, 2010.

\bibitem[Wu17]{YW-ITSTATS}
Yihong Wu.
\newblock {Lecture Notes on Information-theoretic Methods For High-dimensional
  Statistics}.
\newblock http://www.stat.yale.edu/~yw562/teaching/it-stats.pdf, April 2017.

\bibitem[WWRL10]{wang2010privacy}
Cong Wang, Qian Wang, Kui Ren, and Wenjing Lou.
\newblock Privacy-preserving public auditing for data storage security in cloud
  computing.
\newblock In {\em Infocom, 2010 proceedings ieee}, pages 1--9. Ieee, 2010.

\bibitem[ZDW13]{JMLR:v14:zhang13b}
Yuchen Zhang, John~C. Duchi, and Martin~J. Wainwright.
\newblock Communication-efficient algorithms for statistical optimization.
\newblock {\em Journal of Machine Learning Research}, 14:3321--3363, 2013.

\bibitem[ZDW15]{zhang2015divide}
Yuchen Zhang, John Duchi, and Martin Wainwright.
\newblock Divide and conquer kernel ridge regression: A distributed algorithm
  with minimax optimal rates.
\newblock {\em J. Mach. Learn. Res}, 16:3299--3340, 2015.

\end{thebibliography}

\end{document}